\documentclass[10pt,twoside,twocolumn,letterpaper]{article}
\usepackage[letterpaper,textwidth=6.75in,textheight=9.25in,centering]{geometry}
\usepackage{float}
\usepackage[font=small]{caption}
\AtBeginDocument{\fontsize{10}{11}\selectfont}
\makeatletter
\renewcommand{\@maketitle}{%
  \newpage\null\vskip .5em%
  \begin{center}%
    {\@title\par}%
    \vskip 1em%
    {\normalsize\lineskip .4em%
      \begin{tabular}[t]{c}\@author\end{tabular}\par}%
  \end{center}%
  \par\vskip .6em%
}
\renewcommand\section{\@startsection{section}{1}{\z@}%
  {-2.2ex \@plus -.6ex \@minus -.2ex}%
  {1.1ex \@plus .2ex}%
  {\normalfont\large\bfseries\raggedright\hyphenpenalty=10000}}
\renewcommand\subsection{\@startsection{subsection}{2}{\z@}%
  {-1.8ex \@plus -.4ex \@minus -.2ex}%
  {.8ex \@plus .2ex}%
  {\normalfont\normalsize\bfseries}}
\makeatother

\usepackage{amsmath,amssymb,amsthm,mathtools,booktabs,graphicx,microtype,array,balance}
\usepackage[hidelinks,breaklinks]{hyperref}
\hypersetup{pdftitle={The Missing Corner: When Sharp Causal Bounds Fail to Compose},pdfauthor={Yicheng Qi and Xiyi Xiong},pdfsubject={Causal partial identification}}
\newtheorem{theorem}{Theorem}
\newtheorem{proposition}[theorem]{Proposition}
\newtheorem{corollary}[theorem]{Corollary}

\newcommand{\TV}{\operatorname{TV}}
\newcommand{\E}{\mathbb E}

\title{\Large\bfseries The Missing Corner: When Sharp Causal Bounds Fail to Compose}
\author{%
Yicheng Qi\textsuperscript{1}, Xiyi Xiong\textsuperscript{1}\\
\normalfont\textsuperscript{1}Imperial College London, London, United Kingdom\\
\normalfont yq225@ic.ac.uk, xx725@ic.ac.uk
}
\date{}
\begin{document}
\maketitle
\begin{abstract}
Sharp bounds on two potential-outcome means need not combine into sharp bounds on their difference. We characterize exactly when composition is guaranteed across environments sharing a joint potential-outcome law. For binary treatment and outcome, fix strict selection-probability bounds and interior treatment propensities. Subtracting the full model's sharp mean intervals gives sharp average-treatment-effect bounds for every compatible observed law if and only if the absolute selection bounds coincide across environments. Under this alignment, the complete mean region is rectangular; its endpoints and attaining causal models require $O(K)$ arithmetic operations for $K$ environments. For every mismatch, failure occurs on a relatively open set of positive observed tables at the prescribed propensities and persists under small joint drift. We give computable perturbation certificates and an exact family whose gap is first order in the mismatch. A separate finite-confounding family shows when the lost joint information changes an identified causal sign. These results distinguish local marginal sharpness from joint attainability after combining environments.
\end{abstract}

\section{Introduction}
For an intervention with a binary outcome, an analyst wants to compare two population averages: the outcome if everyone were treated and the outcome if everyone were untreated. Their difference is the average treatment effect (ATE). Each unit reveals only the outcome under the treatment it actually receives. Data and assumptions can therefore leave several possible descriptions of the unobserved outcomes. Rather than selecting one description, a bounding analysis asks how large or small each average can be across all admissible descriptions \cite{manski}.

A bound is \emph{sharp} when some admissible model attains it. This creates a tempting calculation: subtract the smallest possible control mean from the largest possible treatment mean. Both numbers are individually attainable. But the model attaining one number may not attain the other. The subtraction is sharp only if a \emph{single} model realizes both extremes. Figure~\ref{fig:core} visualizes this distinction. Every shaded point is a mean pair that one model can realize; the surrounding rectangle records the two separate ranges. In the right panel, its upper-left corner combines the smallest control mean and largest treatment mean. It is separately plausible but jointly impossible.

We study this question when several environments provide treatment--outcome observations. They share the same distribution of responses to treatment, but treatment assignment can differ. To limit selection into treatment, each environment specifies lower and upper treatment probabilities for units with each possible pair of outcomes under control and treatment. We call that pair a \emph{response type}. One environment under these restrictions permits the two extremes to be attained together \cite{dg,mpr}. Combining environments adds a requirement: the \emph{same} response distribution must explain them all. When does that requirement preserve simultaneous attainment?

\begin{figure*}[t]
\centering\includegraphics[width=.88\textwidth]{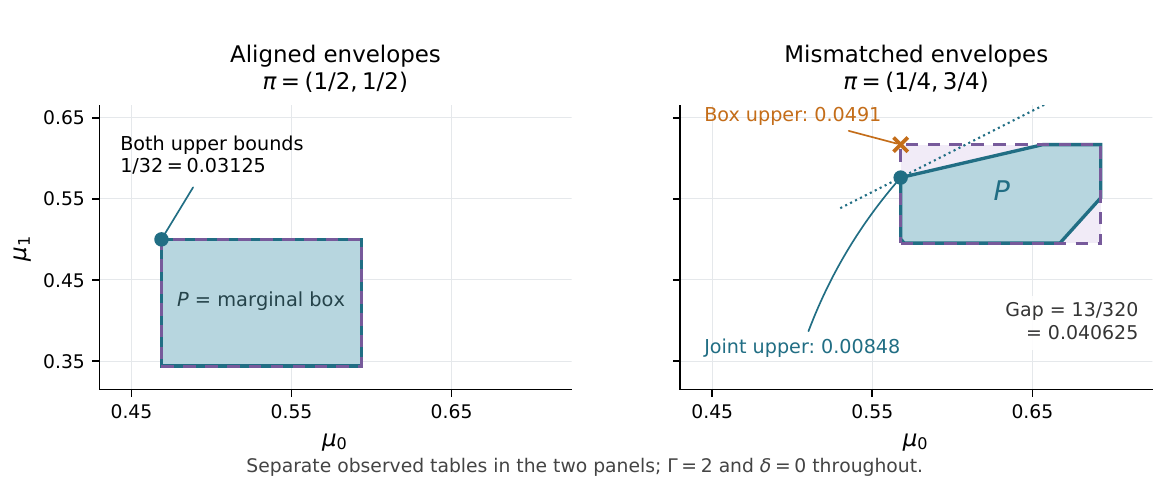}
\caption{\textbf{Two attainable marginals can leave a missing corner.} Each shaded point gives the control mean (horizontal) and treatment mean (vertical) that one causal model can attain together. Dashes show the rectangle formed by their separate ranges. The upper-left corner maximizes their difference in that rectangle. The aligned zero-drift region contains this corner. In the nonaligned example it is excluded, and subtraction exceeds the sharp ATE upper bound by $13/320$. Both panels use $\Gamma=2$ and strictly positive, different observed tables. Rational half-plane intersection verifies all vertices and six extrema. This illustrates a design theorem, not an intervention on fixed data.}
\label{fig:core}
\end{figure*}

Our main result gives an exact answer for binary outcomes, a shared response distribution, and strict selection limits. Subtraction is sharp for every compatible observed collection precisely when the absolute lower and upper selection probabilities are identical across environments. We call this \emph{alignment}. It makes the jointly attainable mean region a rectangle. Every mismatch admits positive observations with a missing corner, although other observations at that same design can compose exactly. The theorem classifies a design's guarantee; the actual observed table determines whether its corner is missing.

The result gives a direct calculation rule. Under alignment, explicit formulas and attaining models take $O(K)$ arithmetic operations for $K$ environments. For an unaligned design, first optimize the four marginal endpoints over the complete model, then check whether each ATE corner has a common realization. An excluded corner calls for joint optimization of that endpoint. Beyond the classification, we construct observations with identical marginal intervals but different population sign conclusions. The appendices quantify the loss near alignment and certify its persistence under observation perturbations and small joint drift. Section~\ref{sec:information} examines finite-sample inference on the stated sign families.

\section{Model and the composition question}
\label{sec:model}
\subsection{Potential outcomes and compatible models}
For a unit, let $A$ denote its observed treatment and $Y$ its observed outcome. The potential outcome $Y(0)$ is its outcome under control and $Y(1)$ its outcome under treatment \cite{rubin}. Standard introductions develop this counterfactual formulation and its assignment assumptions \cite{pearl,imbensrubin,hernan}. Consistency, $Y=Y(A)$, connects these counterfactual quantities to the observed outcome. With a binary outcome, the pair $T=(Y(0),Y(1))$ has four response types: $00$ and $11$ have the same outcome under either treatment, $01$ changes from zero to one under treatment, and $10$ from one to zero. The average treatment effect is $\E[Y(1)]-\E[Y(0)]$. It depends only on the two means, although restrictions on which units receive treatment can depend on their whole response type.

We observe treatment--outcome distributions in several environments. The initial model holds the response-type distribution fixed across environments while allowing treatment assignment to vary. An environment is therefore a source of additional constraints on the same latent population, rather than an independent source of potential-outcome means to subtract. The shared law fixes the joint counterfactual pair, a stronger requirement than sharing its two means, while allowing treatment rates and observed outcome distributions to vary.

A collection of observations is \emph{compatible} if at least one latent model reproduces all of it and obeys the specified restrictions. Typically several models do. Searching over these models gives the range of possible causal means; an endpoint is \emph{sharp} when the search finds a model attaining it. This is \emph{partial identification}: the observations and assumptions determine a range of effects \cite{manski03,tamer,molinari}. Even with unlimited data, that range can contain several values because the unobserved potential outcomes remain unspecified. Throughout the classification we work with known population distributions. This separates uncertainty about unobserved counterfactual outcomes from uncertainty due to a finite sample. Section~\ref{sec:information} considers the latter.

\subsection{Selection restrictions and observed cells}
Index the $K$ environments by $e$, including target $0$. In each environment, the observed data form a $2\times2$ table of treatment $A\in\{0,1\}$ and outcome $Y\in\{0,1\}$. Its cell probabilities and overall treatment probability are
\[
 p_e(a,y)=\Pr_e(A=a,Y=y),\qquad
 \pi_e=\sum_y p_e(1,y).
\]
Write $p=(p_e)_e$ for the collection of observed tables. Initially the environments share the law of $T=(Y(0),Y(1))$. In type order $00,01,10,11$, let
\[
 q(t)=\Pr(T=t),\qquad z_e(t)=\Pr_e(A=1,T=t).
\]
The vector $q$ describes the four latent response proportions. The mass $z_e(t)$ is the proportion both belonging to type $t$ and receiving treatment in environment $e$. On a positive-mass type, its treatment probability is $z_e(t)/q(t)$. Specified limits $\ell_e,u_e$ bound this probability. Writing the restriction in terms of masses gives linear inequalities, while consistency connects those masses to the observed cells:
\begin{align}
 &q\geq0,\quad \sum_tq(t)=1,\quad
 \ell_eq(t)\leq z_e(t)\leq u_eq(t),\label{eq:score}\\
 &p_e(1,y)=\sum_{t:t_1=y}z_e(t),\notag\\
 &p_e(0,y)=\sum_{t:t_0=y}\{q(t)-z_e(t)\}.\label{eq:obs}
\end{align}
Selection mechanisms vary freely within the stated restrictions. Every feasible mass vector has a causal realization: draw $T\sim q$, then assign treatment in environment $e$ with probability $z_e(t)/q(t)$ on positive-mass types. We call this conditional probability the \emph{score}. Thus the mass constraints characterize attainable causal models exactly.

The overall treatment probability $\pi_e$, called the \emph{propensity}, is the average score \cite{rrps}; the two limits bound scores within every response type. In the observation equations, treated units reveal the second coordinate of $T$, and untreated units reveal the first. For example, the cell for treated units with outcome one combines the treated masses of types $01$ and $11$. Different environments observe different mixtures of the same $q$, so their equations must hold together before taking extrema.

\subsection{Marginal sharpness and simultaneous attainment}
The counterfactual population means are $m_0=\E[Y(0)]=q_{10}+q_{11}$ and $m_1=\E[Y(1)]=q_{01}+q_{11}$. Let $P(p)$ contain every pair $(m_0,m_1)$ realized by a model satisfying all environments' equations and restrictions. This is the shaded region in Figure~\ref{fig:core}. Its attained coordinate extrema $L_a=\min_Pm_a$ and $U_a=\max_Pm_a$, $a=0,1$, give the separate mean intervals. We compare the actual range of differences with the interval obtained by subtracting those ranges:
\begin{align}
 I(p)&=\{m_1-m_0:m\in P(p)\},\notag\\
 B(p)&=[L_1-U_0,\ U_1-L_0].\label{eq:comparator}
\end{align}
The upper composition gap is
\begin{equation}\label{eq:gap}
 G(p)=U_1-L_0-\max_{m\in P(p)}(m_1-m_0)\geq0.
\end{equation}
An attainable corner $(L_0,U_1)$ makes the upper bound sharp. The opposing corner $(U_0,L_1)$ does the same for the lower bound. A rectangle is sufficient, but a particular nonrectangular region can still attain both ATE corners.

As a geometric illustration, suppose the attainable pairs lie on the segment from $(0.3,0.5)$ to $(0.5,0.7)$. Every pair has ATE $0.2$, but subtracting the control interval $[0.3,0.5]$ from the treatment interval $[0.5,0.7]$ gives $[0,0.4]$. The corner calculation combines two different feasible points. In our model, all four endpoints already use every environment. The issue is whether the model attaining $L_0$ can also attain $U_1$. Parallel lines of constant $m_1-m_0$ visualize the optimization: subtraction selects the surrounding rectangle's upper-ATE corner, while joint optimization moves the line inward until it touches $P(p)$. A missing corner creates a positive gap even though neither marginal interval can be tightened.

\subsection{Absolute envelopes and a sensitivity parameter}
For the symmetric pair-selection sensitivity model, a parameter $1\leq\Gamma_e<\infty$ gives
\begin{equation}\label{eq:msm}
 \ell_e=\frac{\pi_e}{\pi_e+\Gamma_e(1-\pi_e)},\qquad
 u_e=\frac{\Gamma_e\pi_e}{1-\pi_e+\Gamma_e\pi_e}.
\end{equation}
Odds-ratio restrictions parameterize departures from unconfoundedness \cite{tan,zhao}. The selection bounds apply conditional on the \emph{whole pair} $(Y(0),Y(1))$. Bounds conditional on each potential outcome separately are weaker \cite{mpr}. For example, types $01$ and $11$ share $Y(1)=1$. A restriction conditional on $Y(1)$ averages over these types; a pair-level restriction constrains each separately. Following \cite{mpr}, we call \eqref{eq:msm} the \emph{joint sensitivity model} (JSM); the outcome-specific \emph{marginal sensitivity model} (MSM) conditions on one potential outcome at a time. Theorem~\ref{thm:universal} separates their composition behavior. Absolute limits can also be primitive inputs. Our design theorem fixes $0<\ell_e<\pi_e<u_e<1$ before varying the outcome cells; sample propensities are estimated in the confidence construction of Section~\ref{sec:information}.

To interpret $\Gamma_e$, write $\operatorname{odds}(s)=s/(1-s)$. Equation~\eqref{eq:msm} limits the ratio of a type's treatment odds to the population treatment odds to $[1/\Gamma_e,\Gamma_e]$. It specifies a relative departure from the environment's propensity. Alignment instead concerns the absolute interval of allowable scores. For example, the common value $\Gamma=2$ gives $(\ell,u)=(1/7,2/5)$ at $\pi=1/4$, but $(3/5,6/7)$ at $\pi=3/4$. The relative sensitivity parameter is identical; the permitted scores are not. This distinction is what the classification tests.

The shared-law model also has a familiar instrumental-variable (IV) interpretation: the environment changes treatment assignment while leaving the response-type distribution fixed. Sharp binary-IV bounds can be derived by optimizing response-type masses \cite{bp}. The unrestricted binary IV model has a rectangular mean region \cite{rr}. Dependence between attainable marginal quantities can prevent composition in categorical IV models \cite{song}. Here the new restrictions act on assignment within each joint response type. Theorem~\ref{thm:universal} identifies exactly which finite assignment bounds preserve composition for every observation law compatible with the constrained binary IV model.

\begin{samepage}
\section{When marginal subtraction is exact}
\label{sec:universal-main}
Call a design \emph{aligned} if $(\ell_e,u_e)$ is identical across environments.

\begin{theorem}[Universal composition]\label{thm:universal}
Fix $K\geq2$ and $0<\ell_e<\pi_e<u_e<1$ for each environment. The envelopes are aligned if and only if $P(p)$ is a rectangle for every compatible observed collection at these propensities, and if and only if $I(p)=B(p)$ for every such collection.
Every nonaligned design admits strictly positive compatible tables at the prescribed propensities with $G(p)>0$.
If the same limits instead bound selection conditional on each $Y(a)$ separately, the corresponding mean region is rectangular for every compatible observed collection, without requiring alignment.
\end{theorem}
\end{samepage}

The theorem ranges over observations at a \emph{fixed design}: alignment guarantees composition in every compatible dataset, while every mismatch admits a positive-table counterexample. For a supplied dataset, Section~\ref{sec:workflow} gives the exact attainability test.

For the symmetric JSM with every $1<\Gamma_e<\infty$, this means that universal composition holds if and only if all propensities and all sensitivity parameters are equal. The outcome-specific comparison follows by intersecting the binary IV rectangle with coordinate-wise selection restrictions \cite{rr,mpr}; Appendix~\ref{app:uniform-mismatch} gives the argument. This contrast locates the composition boundary in the restrictions within joint response types.

The log-odds midpoint and half-width of the absolute envelope are $\operatorname{logit}\pi_e$ and $\log\Gamma_e$, respectively. Strictness matters: one environment with $\Gamma_e=1$ identifies both common means, and the unrestricted-IV envelope $[0,1]$ has classical rectangularity \cite{rr}. A propensity at $\pi_e=\ell_e$ or $u_e$ also identifies both means because its score is constant. These cases are outside the classification, which requires $\ell_e<\pi_e<u_e$.

\subsection{Why the boundary is about dependence}
To test a proposed mean pair, we need to recover the response proportions that could produce it. The two means specify how often each outcome is one, but not how often \emph{both} are one. Let $c=q_{11}$ be that remaining proportion. Types $10$ and $01$ then have masses $m_0-c$ and $m_1-c$, and the mass of type $00$ follows by subtraction from one. Thus the entire joint law becomes
\begin{equation}\label{eq:fiber-q}
 q(m,c)=(1-m_0-m_1+c,\ m_1-c,\ m_0-c,\ c).
\end{equation}
Nonnegative masses require $\max(0,m_0+m_1-1)\leq c\leq\min(m_0,m_1)$. For a mean pair satisfying an environment's local mean bounds, its observations and score limits narrow this to an interval $J_e(m)$. We call it a \emph{fiber}: the remaining possibilities for the latent law at a fixed mean pair. The same $c$ must work in every environment. If one environment requires more type-$11$ mass than another allows, no common law realizes that pair. This is the information the two marginal intervals do not retain.

The following construction makes those intervals explicit. Decompose a local law into nonnegative groups assigned at the two endpoint scores,
\[
 q_e^{\rm low}=\frac{u_eq-z_e}{u_e-\ell_e},\qquad
 q_e^{\rm high}=\frac{z_e-\ell_eq}{u_e-\ell_e}.
\]
Within each group, its fixed total mass and two outcome marginals determine a range of allowable $11$ masses, called a Fr\'echet interval \cite{nelsen}. If the two outcome-one masses add to more than the group's total mass, their excess must overlap; that overlap cannot exceed the smaller marginal mass. These give the lower and upper limits. Adding the two group intervals gives the exact fiber for means satisfying the local bounds. Write
\[
 p_e=\begin{pmatrix}1-\pi_e-a_e&a_e\\\pi_e-b_e&b_e\end{pmatrix},
\]
with rows $A$ and columns $Y$, both ordered $0,1$. Put $D_e=u_e-\ell_e$, $r_e=a_e-b_e+\pi_e=\Pr_e(Y\ne A)$, and $s_e=a_e+b_e=\Pr_e(Y=1)$. Then
\begin{equation}\label{eq:local-fiber}
 J_e(m)=[\max\{f_-,A_e,B_e\},\ \min\{f_+,C_e,H_e\}],
\end{equation}
where $f_-=\max\{0,m_0+m_1-1\}$ and $f_+=\min\{m_0,m_1\}$. Here $A_e,B_e$ are additional lower requirements on $c$, and $C_e,H_e$ are upper requirements. The maximum collects all lower requirements; the minimum collects all upper requirements. A common law exists at the proposed means precisely when the resulting intervals overlap across environments. The four affine pieces are
\begin{align*}
 A_e&=\{r_e-(1-u_e)m_0+u_em_1-u_e\}/D_e,\\
 B_e&=\{-r_e+(1-\ell_e)m_0-\ell_em_1+\ell_e\}/D_e,\\
 C_e&=\{s_e-(1-u_e)m_0-\ell_em_1\}/D_e,\\
 H_e&=\{-s_e+(1-\ell_e)m_0+u_em_1\}/D_e.
\end{align*}
Appendix~\ref{sec:aligned-envelope} derives all four affine pieces and the group construction. Under alignment, every cross-environment lower--upper comparison cancels one mean coordinate. Outside alignment an oblique comparison can survive. This is the geometric mechanism behind both directions of Theorem~\ref{thm:universal}.

\subsection{Alignment gives bounds and attaining models}
\label{sec:aligned-formulas}
Why does alignment remove the missing corner? Each environment gives lower and upper requirements on the same $c$. Comparing a lower requirement from one environment with an upper requirement from another combines their information. With common score limits, one mean's coefficient cancels: each resulting requirement bounds only $m_0$ or only $m_1$. The allowable combinations therefore fill a rectangle. The endpoint formulas below record these coordinate-wise restrictions.

For a common envelope $(\ell,u)$, let $r_- = \min_e r_e$, $r_+ = \max_e r_e$, and similarly for $s_-,s_+$. Considering each environment alone gives a local mean rectangle. Their intersection has endpoints $\widetilde L_a,\widetilde U_a$, with explicit formulas in Appendix~\ref{sec:aligned-envelope}. The additional terms below enforce that these local possibilities share one response law. The full-model endpoints are
\begin{align*}
 L_0&=\max\{\widetilde L_0,(r_++s_+-u)/(2-u-\ell)\},\\
 U_0&=\min\{\widetilde U_0,(r_-+s_--\ell)/(2-u-\ell)\},\\
 L_1&=\max\{\widetilde L_1,(s_+-r_-+\ell)/(u+\ell)\},\\
 U_1&=\min\{\widetilde U_1,(s_--r_++u)/(u+\ell)\}.
\end{align*}
They include cross-environment compatibility, rather than only local interval intersections. For example, $A_e\le C_f$ reduces to $(u+\ell)m_1\le s_f-r_e+u$; the other comparison types likewise constrain just one coordinate. Scanning the extrema takes $O(K)$ arithmetic operations.

For any pair in the resulting rectangle, choose $c\in\bigcap_eJ_e(m)$. Split $c$ between each environment's low- and high-score groups, whose binary coupling intervals have that sum. The resulting nonnegative measures reconstruct $q$ and $z_e=\ell q_e^{\rm low}+u q_e^{\rm high}$ and reproduce every observed table. This provides an attaining model at every mean pair, including zero-mass groups. Appendix~\ref{sec:aligned-envelope} gives the complete reconstruction. These formulas explicitly solve the same three-variable model used for an unaligned design.

\subsection{Every mismatch can exclude the ATE corner}
With unequal score limits, that cancellation can fail: a comparison can trade one mean against the other. A sloping boundary may allow a smaller control mean at one point and a larger treatment mean at another, yet exclude their combination. To prove necessity, we must make this boundary active in valid positive observations at the prescribed treatment rates. Appendix~\ref{sec:universal-necessity} constructs such observations for every mismatch; the following sketch explains how it separates the corner.

Parameterize the envelope slopes by
\[
 \alpha_e=(1-u_e)/D_e,\qquad \beta_e=u_e/D_e.
\]
An envelope is determined by $(\alpha_e,\beta_e)$. For a mismatched pair, the construction uses a positive common type law, endpoint scores in one environment, and scores in the other chosen to maintain its exact propensity. Additional environments use their interior constant scores. Exactly one cross-fiber comparison is active at a mean point $m^*$; all other comparisons and local mean constraints have slack. That comparison is
\begin{equation}\label{eq:oblique}
 -a(m_0-m_0^*)+b(m_1-m_1^*)\leq0,\qquad a,b>0.
\end{equation}
Moving in direction $(b,a)$ leaves the active comparison at equality: one mean's increase offsets the other's. For sufficiently small $t>0$, both $m^*\pm t(b,a)$ are feasible. Hence the separately attained corner $(L_0,U_1)$ violates the inequality by at least $2abt$. For every feasible $m$,
\[
 (U_1-L_0)-(m_1-m_0)\geq2\min(a,b)t>0.
\]
The two nearby feasible points supply the small control mean and large treatment mean separately; the active boundary prevents combining them. This proves the converse with positive observations at the fixed propensities.

\subsection{Checking the observed table}
\label{sec:workflow}
The theorem answers which designs make subtraction safe for every compatible table. For a particular table at an unaligned design, the question is more direct: do its two extremal mean pairs have a common response law? The calculation guide gives the steps, and the fiber intervals supply the check.

\begin{table}[H]
\centering\small
\begin{tabular}{@{}p{.22\columnwidth}p{.72\columnwidth}@{}}
\toprule
Design & Calculation for sharp ATE bounds \\
\midrule
Aligned, zero drift & Use the explicit marginal formulas and attaining models in Section~\ref{sec:aligned-formulas}; subtraction is sharp. \\
Unaligned, zero drift & Solve four marginal objectives over the three-variable model. Check both corners; optimize each joint ATE endpoint whose check fails. \\
Positive drift & Optimize marginal and ATE objectives over the target/source mass model with the drift constraints. \\
\bottomrule
\end{tabular}
\par\smallskip\noindent\textbf{Calculation guide.} The aligned formulas compute endpoints and witnesses in $O(K)$ arithmetic. For an unaligned table, the $O(K)$ corner scan follows the four marginal optimizations. In every row, subtraction contains the sharp ATE interval.
\end{table}

For means in all local rectangles, define
\[
\begin{split}
 \underline c(m)&=\max\{0,m_0+m_1-1,\max_e A_e,\max_e B_e\},\\
 \overline c(m)&=\min\{m_0,m_1,\min_e C_e,\min_e H_e\}.
\end{split}
\]
The upper subtraction endpoint is sharp exactly when
\begin{equation}\label{eq:corner-test}
 \underline c(L_0,U_1)\le\overline c(L_0,U_1).
\end{equation}
The lower endpoint uses the same test at $(U_0,L_1)$. Necessity follows because attaining $U_1-L_0$ forces both means to equal their marginal extrema. Sufficiency follows from the common-$c$ construction. Once the marginal endpoints are available, each check scans $O(K)$ affine values and returns an attaining model or two conflicting fiber pieces.

At positive drift, environments can have different dependence coordinates. The calculation then uses the full drift model in \eqref{eq:drift}.

The check is reproducible through \path{research24/diagnose.py}. For rational population tables, the program uses aligned formulas when available. Otherwise, it locates LP bases in $(m_0,m_1,c)$ and verifies feasibility, dual signs and equality of primal and dual objectives in rational arithmetic. It outputs the two corner tests, conflict certificates or attaining masses, and the resulting sharp ATE endpoints. Appendix~\ref{app:workflow-records} gives the input format, certification-status handling and verification records.

For Figure~\ref{fig:core}'s mismatched tables, the upper corner is $(109/192,829/1344)$. Evaluating its fiber gives
\[
 A_0=109/192>59/144=C_1.
\]
The corner would require the same proportion $c$ to be at least $109/192$ and at most $59/144$. These conflicting requirements exclude it directly. Written in terms of the means, their comparison is
\[
 -16m_0+35m_1\le133/12.
\]
\fontsize{10}{12.7}\selectfont
Both means are individually sharp. Joint optimization gives $U_\tau=19/2240$, versus subtraction's $11/224$, a loss of $13/320$. The fiber separation $91/576$ measures the conflicting requirements on the common response mass $c$.

The same unaligned $\Gamma=2$, $\pi=(1/4,3/4)$ design can instead compose exactly. With tables
\[
 p_0=\frac18\begin{pmatrix}3&3\\1&1\end{pmatrix},\qquad
 p_1=\frac18\begin{pmatrix}1&1\\3&3\end{pmatrix},
\]
both sharp mean intervals are $[7/16,9/16]$, and both ATE corners pass: $I=B=[-1/8,1/8]$. Thus a mismatch identifies where the universal shortcut ends, while the fibers decide whether subtraction actually loses sharpness.

\balance
\section{Population sign loss and sampling}
\label{sec:information}
A negative ATE upper bound certifies a negative effect in every compatible model. The following family holds all four marginal endpoints fixed while the joint upper bound changes sign. It shows what information is lost by replacing the attainable mean region with its two coordinate intervals. Here we also allow the response-type law to vary slightly across environments. This \emph{drift} is measured relative to the target law $q_0$, under which the target means are taken:
\begin{equation}\label{eq:drift}
 \TV(q_e,q_0)=\tfrac12\|q_e-q_0\|_1\le\delta_e,\qquad e>0.
\end{equation}

Total variation here measures the probability mass that must be moved between response types to turn one law into the other. A budget $\delta_e=0.01$ permits moving at most one percent of that mass. At zero budget the laws coincide, recovering the shared-law model.

\begin{proposition}[Exact decision loss under joint drift]\label{thm:insufficiency}
With $\Gamma_e=2$, $\pi=(1/4,3/4)$, $t\in[-1,1]$, and $\delta=\delta_1\in[0,1/100]$, let
\begin{align*}
 p_0(t)&=\frac1{160}\begin{pmatrix}60&60\\25-t&15+t\end{pmatrix},\\
 p_1&=\frac1{160}\begin{pmatrix}27&13\\74&46\end{pmatrix}.
\end{align*}
The sharp target mean endpoints are $(7/16,9/16,161/480-\delta,221/480+\delta)$, independent of $t$, while
\begin{equation}\label{eq:summary-family}
 U_\tau(p(t),\delta)=t/160+\delta.
\end{equation}
The laws at $t=\pm1$ have different negativity decisions for every $\delta<1/160$. The distinction persists on paired observation neighborhoods for $\delta\le1/320$.
\end{proposition}

The common subtraction interval is $[-109/480-\delta,11/480+\delta]$. At zero drift, the joint upper bounds are $-1/160$ for $t=-1$ and $1/160$ for $t=1$, although subtraction gives the same positive upper bound $11/480$. The first table certifies a negative effect in every compatible model; the second admits a positive effect. This distinction persists whenever $\delta<1/160$. Appendix~\ref{app:decision-information} supplies dual identities and attaining causal models. Appendix~\ref{sec:mirror-design-information} gives the separate mirror class.

\label{sec:evidence}
The population calculation treats the observed cell probabilities as known. A finite sample instead supplies uncertain estimates of them. This identification--inference distinction underlies confidence intervals for partially identified parameters \cite{im,stoye}. Both procedures therefore optimize the same confidence-based outer causal model $\mathcal F_n$: the joint upper maximizes $m_1-m_0$ over it, while subtraction uses $\max_{\mathcal F_n}m_1-\min_{\mathcal F_n}m_0$. A negativity declaration requires this entire outer range to lie below zero. Sampling uncertainty enlarges the set of plausible models, so a negative population upper bound can become a nonnegative confidence upper bound. The confidence construction, frozen protocols and Monte Carlo intervals are in Appendix~\ref{app:revision-evidence}.

\begin{table}[H]
\centering\small
\begin{tabular}{lccc}
\toprule
Family setting & Pop. sep. & Joint & Subtr.\\
\midrule
$t=-1,\ \delta=0$ & Yes & 0,0,0 & 0,0,0\\
$t=1,\ \delta=0$ & No & 0,0,0 & 0,0,0\\
$t=-1,\ \delta=1/320$ & Yes & 0,0,0 & 0,0,0\\
$t=1,\ \delta=1/320$ & No & 0,0,0 & 0,0,0\\
Mirror $t=-\epsilon/2$ & Yes & 0,0,0 & 0,0,0\\
Mirror $t=\epsilon/2$ & No & 0,0,0 & 0,0,0\\
\bottomrule
\end{tabular}
\caption{\textbf{Sampling the stated sign families.} Each count triple corresponds to $n=3000,10000,30000$ per environment, with 200 trials at each size. ``Pop. sep.'' means the population joint upper is negative but subtraction's is nonnegative. The first four rows use Proposition~\ref{thm:insufficiency}; the mirror rows use Theorem~\ref{thm:design-information} in Appendix~\ref{sec:mirror-design-information}, with $(l,u,p)=(1/7,2/5,1/4)$ and $\epsilon=\epsilon_0/2$. Both procedures use the same outer model. All six settings were fixed before sampling.}
\label{tab:abstract-sampling}
\end{table}

Table~\ref{tab:abstract-sampling} reports all six frozen settings. Both procedures make zero negativity declarations in the 3600 trials, including the population-separated negative-$t$ settings. Under the specified outer confidence model, these constructions therefore show no observed sample-level sign-detection advantage. The common feasible set controls the information supplied to both procedures; separating confidence-region conservatism from intrinsic sampling difficulty remains an open statistical question.

\label{sec:main-end}
\clearpage\nobalance
\fontsize{10}{11}\selectfont

\clearpage\onecolumn
\section*{Reproducibility Checklist}
\begin{enumerate}
\item For all models and algorithms presented, check if you include:
\begin{enumerate}
\item A clear description of the mathematical setting, assumptions, algorithm, and/or model.
[Yes] Section~\ref{sec:model} defines the models, common-law assumption and full-model comparator. Section~\ref{sec:workflow} gives the calculation, and Section~\ref{sec:information} defines joint drift.
\item An analysis of the properties and complexity (time, space, sample size) of any algorithm.
[Yes] Theorem~\ref{thm:rectangle} in Appendix~\ref{sec:aligned-envelope} gives the $O(K)$ arithmetic construction. Section~\ref{sec:workflow} and Appendix~\ref{app:workflow-records} specify the marginal optimizations, corner checks and joint optimizations for failed corners. No new general solver-runtime or statistical sample-complexity guarantee is claimed; finite-sample methods and costs are documented in Appendix~\ref{app:revision-evidence} and the source supplement.
\item (Optional) Anonymized source code, with specification of all dependencies, including external libraries.
[Yes] The anonymous supplement contains source code, checked dependency versions, saved inputs, certificates, and reproduction commands.
\end{enumerate}

\item For any theoretical claim, check if you include:
\begin{enumerate}
\item Statements of the full set of assumptions of all theoretical results.
[Yes] The statements specify strict envelopes, binary outcomes, common or drifting joint laws, fixed propensities, and the additional conditions for perturbation certificates.
\item Complete proofs of all theoretical results.
[Yes] The included theoretical appendices give the constructive proofs, separating inequalities and attaining witnesses. Appendix~\ref{app:revision-evidence} supplies the outer-confidence-model coverage argument. Numerical examples do not replace universal proofs.
\item Clear explanations of any assumptions.
[Yes] Sections~\ref{sec:model} and \ref{sec:evidence} explain their meaning and limits; Appendix~\ref{sec:robust-universal-failure} states the perturbation assumptions. Sensitivity and drift budgets are not claimed to be identified from the observed tables.
\end{enumerate}

\item For all figures and tables that present empirical results, check if you include:
\begin{enumerate}
\item The code, data, and instructions needed to reproduce the main experimental results (either in the supplemental material or as a URL).
[Yes] The supplement contains every saved sign-family count table, all reported benchmark rows, protocols, scripts, exact figure geometry, and reproduction instructions.
\item All the training details (e.g., data splits, hyperparameters, how they were chosen).
[Yes] No predictive model is trained. Appendix~\ref{app:revision-evidence}, \path{research23/PROTOCOL.md}, and the archived protocols specify the grids, trial counts, seed, confidence level, solver threshold, and mechanism selection.
\item A clear definition of the specific measure or statistics and error bars (e.g., with respect to the random seed after running experiments multiple times).
[Yes] Counts are reported with trial denominators. Appendix~\ref{app:revision-evidence} and the saved summaries define the declaration rule and supply pointwise binomial Monte Carlo intervals. The theoretical figures are exact population constructions, not sampled error-bar plots. Reused designs and the absence of a prevalence interpretation are explicit.
\item A description of the computing infrastructure used. (e.g., type of GPUs, internal cluster, or cloud provider).
[Yes] Appendix~\ref{app:reproducibility} records CPU computation and the production software versions. The historical CPU model and memory were not recorded; saved experiment records retain the reported elapsed times.
\end{enumerate}

\item If you are using existing assets (e.g., code, data, models) or curating/releasing new assets, check if you include:
\begin{enumerate}
\item Citations of the creator If your work uses existing assets.
[Yes] Prior methods are cited. The numerical libraries are identified in the dependency and third-party notices. There is no third-party empirical dataset or pretrained model.
\item The license information of the assets, if applicable.
[Yes] The supplement identifies software dependencies and their licensing sources, retains third-party notices, and releases author-created materials under Creative Commons Attribution 4.0 International.
\item New assets either in the supplemental material or as a URL, if applicable.
[Yes] Synthetic records, verification programs and exact certificates accompany the submission; the analytic proofs are included in the paper's appendices.
\item Information about consent from data providers/curators.
[Not Applicable] The observations are mathematical examples and programmatically simulated counts; no external individual-level dataset is used.
\item Discussion of sensible content if applicable, e.g., personally identifiable information or offensive content.
[Not Applicable] The anonymous supplement contains synthetic research records, with no individual-level personal data. Author contact information is kept outside that supplement.
\end{enumerate}

\item If you used crowdsourcing or conducted research with human subjects, check if you include:
\begin{enumerate}
\item The full text of instructions given to participants and screenshots.
[Not Applicable] No participants were recruited and no crowdsourcing was conducted.
\item Descriptions of potential participant risks, with links to Institutional Review Board (IRB) approvals if applicable.
[Not Applicable] There was no human-subject experiment. All new observations are synthetic.
\item The estimated hourly wage paid to participants and the total amount spent on participant compensation.
[Not Applicable] No participants were recruited or compensated.
\end{enumerate}
\end{enumerate}

\clearpage\onecolumn\appendix\raggedbottom
The appendices follow the main argument. For the classification proof, read A for the aligned construction and B for the mismatched corner. C--D establish the population sign examples; E quantifies mismatch and verifies the outcome-specific comparison; F--G establish persistence and its computable certificate. H quantifies composition loss under positive drift. For implementation and statistical details, I documents the sampling experiment and J gives the calculation entry point and verification records. The supplement README indexes the code, numerical records and experiment protocols.

\section{Binary rectangularity with an aligned selection envelope}
\label{sec:aligned-envelope}

The construction has three steps: solve each environment separately, express its remaining joint-law freedom as an interval for $c=q_{11}$, and intersect these intervals. Alignment turns the cross-environment comparisons into separate bounds on the two means. The mean region is a projection of a convex polytope, so convex geometry provides the underlying framework \cite{rockafellar}.

Consider finitely many environments with binary treatment and
outcome, a common potential-outcome law $q$, and separate selection
masses $z_e$. The common-law assumption is exactly zero drift.
Suppose the primitive selection restrictions are
\[
 \ell q(t)\leq z_e(t)\leq u q(t)
 \quad\text{for all }e,t,\qquad 0<\ell<u<1,
\]
with the same absolute limits $\ell,u$ in every environment.
The observed propensities need not agree. Write their tables as
\[
 p_e=\begin{pmatrix}1-\pi_e-a_e&a_e\\\pi_e-b_e&b_e\end{pmatrix},
 \qquad \ell\leq\pi_e\leq u,
\]
where all cells are nonnegative. If any $\pi_e$ lies outside this
interval, the model is immediately infeasible by summing the
selection inequalities. No additional source-to-source restrictions
are imposed on the selection witnesses.

Define
\[
 r_e=a_e-b_e+\pi_e=p_e(0,1)+p_e(1,0),\qquad
 s_e=a_e+b_e=p_e(0,1)+p_e(1,1).
\]
Let $r_-,r_+,s_-,s_+$ denote their respective minima and maxima
over all environments, including the target. Define the local
rectangle intersection bounds
\begin{align*}
 L_0&=\max_e\max\left\{\frac{a_e}{1-\ell},
                      1-\frac{1-\pi_e-a_e}{1-u}\right\},&
 U_0&=\min_e\min\left\{\frac{a_e}{1-u},
                      1-\frac{1-\pi_e-a_e}{1-\ell}\right\},\\
 L_1&=\max_e\max\left\{\frac{b_e}{u},
                      1-\frac{\pi_e-b_e}{\ell}\right\},&
 U_1&=\min_e\min\left\{\frac{b_e}{\ell},
                      1-\frac{\pi_e-b_e}{u}\right\}.
\end{align*}

\begin{theorem}[Constructive aligned solution]\label{thm:rectangle}
The full sharp target mean region is
\[
 P=[L_0^*,U_0^*]\times[L_1^*,U_1^*],
\]
where
\begin{align*}
 L_0^*&=\max\left\{L_0,\frac{r_++s_+-u}{2-u-\ell}\right\},&
 U_0^*&=\min\left\{U_0,\frac{r_-+s_--\ell}{2-u-\ell}\right\},\\
 L_1^*&=\max\left\{L_1,\frac{s_+-r_-+\ell}{u+\ell}\right\},&
 U_1^*&=\min\left\{U_1,\frac{s_--r_++u}{u+\ell}\right\}.
\end{align*}
An interval with its lower endpoint above its upper endpoint is
empty. These formulas therefore also characterize feasibility.
If the model is nonempty, all four mean corners are attainable and
the sharp ATE interval is $[L_1^*-U_0^*,U_1^*-L_0^*]$.

Computing these extrema and constructing a compatible witness at any mean pair takes $O(K)$ arithmetic operations.
\end{theorem}

\begin{proof}
Put $D=u-\ell$. For fixed means $m_0,m_1$, any compatible law has
the form
\[
 q=(1-m_0-m_1+c,\ m_1-c,\ m_0-c,\ c)
\]
in type order $00,01,10,11$. Decompose a local selection witness as
\[
 q_e^{\rm low}=(u q-z_e)/D,\qquad
 q_e^{\rm high}=(z_e-\ell q)/D.
\]
The two nonnegative measures have masses
$(u-\pi_e)/D$ and $(\pi_e-\ell)/D$. Their two marginal outcome-one
masses are
\[
 \left(\frac{a_e-(1-u)m_0}{D},
       \frac{u m_1-b_e}{D}\right),\qquad
 \left(\frac{(1-\ell)m_0-a_e}{D},
       \frac{b_e-\ell m_1}{D}\right).
\]
Requiring each marginal to lie between zero and its group's total
mass gives exactly the local rectangle whose endpoints enter
$L_a,U_a$. This remains valid when a group has zero total mass.

Within either group, its binary $11$ mass ranges between the
Fr\'echet lower and upper bounds. Adding these two intervals and
expanding gives the exact local fiber
\[
 J_e(m)=\left[
   \max\{0,m_0+m_1-1,A_e,B_e\},
   \min\{m_0,m_1,C_e,H_e\}\right],
\]
for means in its local rectangle, where
\begin{align*}
 A_e&=\frac{r_e-(1-u)m_0+u m_1-u}{D},&
 B_e&=\frac{-r_e+(1-\ell)m_0-\ell m_1+\ell}{D},\\
 C_e&=\frac{s_e-(1-u)m_0-\ell m_1}{D},&
 H_e&=\frac{-s_e+(1-\ell)m_0+u m_1}{D}.
\end{align*}
In particular each local rectangle is sharp. Local witnesses with
a common $c$ can be assembled independently into the full model.

For means in every local rectangle, intersecting these fibers
requires only the additional comparisons of each $A_e,B_e$ against
each $C_f,H_f$. Comparisons with the common lower pieces
$0,m_0+m_1-1$ or common upper pieces $m_0,m_1$ already hold within
an appropriate local fiber. The four remaining comparisons are
\begin{align*}
 A_e\leq C_f&\iff (u+\ell)m_1\leq s_f-r_e+u,\\
 A_e\leq H_f&\iff (2-u-\ell)m_0\geq r_e+s_f-u,\\
 B_e\leq C_f&\iff (2-u-\ell)m_0\leq r_e+s_f-\ell,\\
 B_e\leq H_f&\iff (u+\ell)m_1\geq s_f-r_e+\ell.
\end{align*}
Both coefficients multiplying a mean are positive. Taking extrema
over all ordered pairs $(e,f)$ yields precisely the four displayed
extra scalar bounds. If they hold, choose any common $c$ in the
fiber intersection; the two-group construction supplies compatible
selection masses for each environment. This proves equality of
the feasible mean region with the asserted rectangle, and linear
optimization over it gives the ATE interval.
\end{proof}

\paragraph{Corollary for the symmetric joint sensitivity model.}
In the finite symmetric pair-selection JSM, if all environments
have common $\pi\in(0,1)$ and common $\Gamma>1$, their selection
limits align as
\[
 \ell=\frac{\pi}{\pi+\Gamma(1-\pi)},\qquad
 u=\frac{\Gamma\pi}{1-\pi+\Gamma\pi}.
\]
Substituting $r_e=(a_e-b_e)+\pi$ recovers the common-propensity
rectangle and its explicit bounds. The log-odds midpoint and half-width identities following
the JSM consequence of Theorem~\ref{thm:universal} show that aligned finite symmetric-JSM limits
force both common propensity and common sensitivity parameter.
The primitive-envelope result above permits varying propensities
without asserting that extension for the symmetric JSM.

\section{Alignment is necessary for universal marginal composition}
\label{sec:universal-necessity}

The sufficiency result in Section~\ref{sec:aligned-envelope} has a
converse when all selection envelopes are strict and noncollapsed.
At any fixed mismatched design, we construct positive observations
with a missing ATE corner. The proof first chooses two unequal
envelopes, then makes one cross-fiber comparison active while all
other constraints have slack. Two small moves along that boundary
provide separately attainable extremes whose combination violates it.

Let $\mathcal E$ be a nonempty finite set of environments, with a
designated target. Treatment and outcome are binary. Fix
\[
 0<\ell_i<u_i<1,\qquad \ell_i<\pi_i<u_i
 \quad(i\in\mathcal E).
\]
These are primitive selection limits and prescribed observed
propensities. At zero drift there is one common response-type law
$q=(q_{00},q_{01},q_{10},q_{11})$, and each environment has its own
selection masses satisfying
\[
 \ell_iq_t\leq z_{it}\leq u_iq_t.
\]
There are no other restrictions coupling the selection witnesses.
For observed tables $p=(p_i)_{i\in\mathcal E}$ with these
propensities, let $P(p)$ be the full feasible region of
\[
 m_0=q_{10}+q_{11},\qquad m_1=q_{01}+q_{11}.
\]
For a compatible $p$, write
\[
 L_0(p)=\min_{m\in P(p)}m_0,\qquad
 U_1(p)=\max_{m\in P(p)}m_1,\qquad
 U_\tau(p)=\max_{m\in P(p)}(m_1-m_0).
\]
All three extrema are attained by compactness.

\paragraph{Theorem (universal composition boundary).}
For the fixed design above, the following statements are equivalent.
\begin{enumerate}
 \item The envelopes are aligned:
       $(\ell_i,u_i)=(\ell_j,u_j)$ for every $i,j$.
 \item For every compatible observed table collection with
       propensities $(\pi_i)$, the full region $P(p)$ is a rectangle.
 \item For every compatible \emph{strictly positive} observed table
       collection with propensities $(\pi_i)$,
       $U_\tau(p)=U_1(p)-L_0(p)$.
\end{enumerate}
More specifically, if any two envelopes differ, there are strictly
positive compatible tables with the prescribed propensities and
\begin{equation}\label{eq:universal-strict-gap}
 U_1(p)-L_0(p)-U_\tau(p)>0.
\end{equation}
The statement includes any finite number of environments. Alignment
is vacuous when there is only one.

\begin{proof}
Aligned-envelope sufficiency is proved in
Section~\ref{sec:aligned-envelope}; it does not require equal
propensities. A nonempty rectangle attains its coordinate corners,
so the first statement implies the second and the second implies
the third. We prove the converse by constructing
\eqref{eq:universal-strict-gap} whenever a distinct pair exists.

\emph{Local fibers with environment-specific limits.}
Put $D_i=u_i-\ell_i$, $a_i=p_i(0,1)$, and $b_i=p_i(1,1)$, so
\[
 p_i=\begin{pmatrix}1-\pi_i-a_i&a_i\\\pi_i-b_i&b_i\end{pmatrix}.
\]
Apply the two-group decomposition in Appendix~\ref{sec:aligned-envelope}
separately with limits $(\ell_i,u_i)$. Its group masses are
$w_i^{\rm low}=(u_i-\pi_i)/D_i$ and
$w_i^{\rm high}=(\pi_i-\ell_i)/D_i$; its group marginal formulas
are the displayed aligned formulas with $\ell,u,D$ replaced by
$\ell_i,u_i,D_i$. Write those marginals as $x_{ai}^v$.
The conditions $0\leq x_{ai}^v\leq w_i^v$ define the local mean
rectangle $R_i$. The binary coupling intervals of the two groups
sum to the exact local dependence fiber
\begin{equation}\label{eq:universal-local-fiber}
 J_i(m)=\left[
 \max\{0,m_0+m_1-1,A_i,B_i\},
 \min\{m_0,m_1,C_i,H_i\}\right],\qquad m\in R_i,
\end{equation}
where
\begin{align*}
 A_i&=\frac{a_i-b_i-(1-u_i)m_0+u_i m_1-u_i+\pi_i}{D_i},&
 B_i&=\frac{-a_i+b_i+(1-\ell_i)m_0-\ell_i m_1-\pi_i+\ell_i}{D_i},\\
 C_i&=\frac{a_i+b_i-(1-u_i)m_0-\ell_i m_1}{D_i},&
 H_i&=\frac{-a_i-b_i+(1-\ell_i)m_0+u_i m_1}{D_i}.
\end{align*}
As in the aligned case, every point of this interval is attainable
by selecting the two group coupling masses. Consequently
$m\in P(p)$ if and only if $m\in\bigcap_iR_i$ and all the intervals
$J_i(m)$ have a common point. This is a finite system of affine
inequalities comparing every lower piece with every upper piece.

\emph{Choice of the pair and the common law.}
Define
\[
 \alpha_i=\frac{1-u_i}{D_i},\qquad
 \beta_i=\frac{u_i}{D_i},\qquad
 \omega_i=\frac{\pi_i-\ell_i}{D_i}\in(0,1).
\]
The pair $(\alpha_i,\beta_i)$ determines the envelope, since
$\alpha_i+\beta_i=1/D_i$, $u_i=\beta_iD_i$, and
$\ell_i=u_i-D_i$. Thus either some pair has different $\alpha$
values, or a pair with equal $\alpha$ has different $\beta$ values.
Choose $e,f$ with $\alpha_e>\alpha_f$ in the first case, and with
$\alpha_e=\alpha_f$, $\beta_e>\beta_f$ in the second.

Choose
\[
 0<\varepsilon<
 \min\{\omega_e,1-\omega_e,\omega_f,1-\omega_f\}
\]
and set
\begin{equation}\label{eq:universal-common-law}
 q^*=(\omega_e-\varepsilon,\ \varepsilon,\ \varepsilon,
                    1-\omega_e-\varepsilon).
\end{equation}
All its coordinates are positive. In environment $e$, use the
selection probabilities
\[
 h_e=\left(u_e,\frac{\ell_e+u_e}{2},
                \frac{\ell_e+u_e}{2},\ell_e\right).
\]
Their $q^*$-weighted average is exactly $\pi_e$. Put
\[
 t_f=\frac{\pi_f-\varepsilon(\ell_f+u_f)}{1-2\varepsilon}.
\]
The denominator is positive, and the two inequalities
$t_f>\ell_f$ and $t_f<u_f$ are equivalent to
$\varepsilon<\omega_f$ and $\varepsilon<1-\omega_f$, respectively.
In the first case choose
\[
 h_f=(t_f,\ell_f,u_f,t_f),
\]
and in the second choose
\[
 h_f=(t_f,u_f,\ell_f,t_f).
\]
Both patterns have $q^*$-weighted average $\pi_f$. Every remaining
environment $i$ uses the constant selection probability $h_i=\pi_i$.
Let $z_{it}=h_{it}q_t^*$ and generate the observed cells from
$(q^*,z_i)$. These tables are compatible, have exactly the
prescribed propensities, and have strictly positive cells because
$q_t^*>0$ and $0<h_{it}<1$.

\emph{A unique active cross comparison.}
Let $m^*$ be the means of $q^*$ and $c^*=q_{11}^*$. For any of
these decomposed witnesses, the pieces in
\eqref{eq:universal-local-fiber}, evaluated at $m^*$, satisfy
\begin{align*}
 A_i&=c^*-(q_i^{\rm high})_{11}-(q_i^{\rm low})_{00},&
 B_i&=c^*-(q_i^{\rm low})_{11}-(q_i^{\rm high})_{00},\\
 C_i&=c^*+(q_i^{\rm low})_{10}+(q_i^{\rm high})_{01},&
 H_i&=c^*+(q_i^{\rm low})_{01}+(q_i^{\rm high})_{10}.
\end{align*}
In environment $e$, the only zero cells of the two groups are low
$00$ and high $11$. Thus $A_e=c^*$, and every other lower piece is
strictly below $c^*$, while every upper piece is strictly above it.
For environment $f$ in the first case, the only zero cells are low
$10$ and high $01$: here $C_f=c^*$ is the unique tight upper piece.
In the second case the zero cells are low $01$ and high $10$, so
$H_f=c^*$ is uniquely tight instead. Every lower piece of $f$ is
strictly below $c^*$. Both groups in every remaining environment
have all cells positive, hence all their fiber comparisons are
strict at $c^*$.

The common pieces are strict as well:
$0<c^*$, $m_0^*+m_1^*-1=c^*-q_{00}^*<c^*$, and
$m_a^*>c^*$ for $a=0,1$. In every group, each outcome marginal and
its complement contain a positive cell, including the groups with
one zero cell. Therefore
$0<x_{ai}^v<w_i^v$: the point $m^*$ is strictly inside every local
rectangle. It follows that the sole active inequality in the full
mean system is $A_e\leq C_f$ in the first case, or $A_e\leq H_f$
in the second.

\emph{Orientation and the upper composition gap.}
The active inequality can be written globally as
\begin{equation}\label{eq:universal-oblique-facet}
 -a(m_0-m_0^*)+b(m_1-m_1^*)\leq0,
 \qquad a,b>0.
\end{equation}
Indeed, its coefficients are
\[
 (a,b)=
 \begin{cases}
 (\alpha_e-\alpha_f,\ \beta_e+\beta_f-1),
       &\alpha_e>\alpha_f,\\
 (\alpha_e+\alpha_f+1,\ \beta_e-\beta_f),
       &\alpha_e=\alpha_f,\ \beta_e>\beta_f.
 \end{cases}
\]
They are positive; in particular $\beta_i>1$ because $\ell_i>0$.
All other affine inequalities are strict at $m^*$, so for some
$t>0$ both $m^*-t(b,a)$ and $m^*+t(b,a)$ belong to $P(p)$.
They remain on \eqref{eq:universal-oblique-facet}, and sufficiently
small displacements preserve every other inequality. Hence
\[
 L_0(p)\leq m_0^*-tb,\qquad
 U_1(p)\geq m_1^*+ta.
\]
The violation of \eqref{eq:universal-oblique-facet} at the marginal
corner $(L_0(p),U_1(p))$ is therefore at least $2abt>0$.
For any feasible $m$, subtracting its valid inequality from this
corner violation gives
\[
 a\{m_0-L_0(p)\}+b\{U_1(p)-m_1\}\geq2abt.
\]
Both braces are nonnegative by definition of the coordinate
extrema. Consequently
\[
 \{U_1(p)-L_0(p)\}-(m_1-m_0)
 \geq\frac{2abt}{\max(a,b)}
 =2\min(a,b)t>0.
\]
Maximizing $m_1-m_0$ proves \eqref{eq:universal-strict-gap}, completing
the contrapositive and the equivalence.
\end{proof}

\paragraph{Corollary (symmetric JSM designs).}
Fix interior propensities $\pi_i\in(0,1)$ and finite sensitivity
parameters $\Gamma_i>1$ in the symmetric pair-selection JSM.
At zero drift with binary outcomes, universal full-model
rectangularity, and equivalently universal upper marginal
composition, holds if and only if all $\pi_i$ agree and all
$\Gamma_i$ agree.

\begin{proof}
The selection endpoints satisfy
\[
 \operatorname{logit}\pi_i
  =\tfrac12\{\operatorname{logit}\ell_i+
             \operatorname{logit}u_i\},\qquad
 \log\Gamma_i
  =\tfrac12\{\operatorname{logit}u_i-
             \operatorname{logit}\ell_i\}.
\]
Thus aligned endpoints are equivalent to common propensity and
common sensitivity parameter. The theorem applies because its
construction preserves the prescribed propensities exactly.
\end{proof}

\paragraph{Scope and boundary cases.}
In particular, any fixed nonzero propensity mismatch at common
$\Gamma>1$, or any fixed nonzero sensitivity mismatch at common
interior propensity, permits strictly positive compatible tables
with a strict upper composition gap. This remains an existence
statement for arbitrarily small mismatches; no positive gap uniform
as the mismatch vanishes is asserted.

The strict noncollapsed assumption cannot simply be dropped. If an
environment has $\Gamma_i=1$, it identifies both common means from
its own observed table. A compatible zero-drift model then has a
singleton mean region, regardless of other envelopes. Propensity
boundaries and other degenerate envelopes require separate
statements.

Finally, the witness \eqref{eq:universal-common-law} has
$q_{01}^*=q_{10}^*$ and hence effect zero. The construction establishes nonrectangularity and a strict composition
gap. Strict sign identification is studied through the families in
Appendices~\ref{sec:mirror-design-information}--\ref{app:decision-information}.

\section{A design class with indistinguishable mean summaries and different signs}
\label{sec:mirror-design-information}

This appendix supplies the explicit observations, attaining models and stability proof for Theorem~\ref{thm:design-information}. The target is environment zero, and the full common-law model imposes no restriction fixing a dependence coordinate or a score beyond its envelope.

\begin{theorem}[Indistinguishable summaries in a mirror class]
\label{thm:design-information}
Fix two environments with $(\ell_0,u_0,\pi_0)=(l,u,p)$ and $(\ell_1,u_1,\pi_1)=(1-u,1-l,1-p)$, where
\[
\begin{gathered}
 0<l<1/5,\quad l<u<9/20,\\
 1/3<(p-l)/(u-l)<4/9.
\end{gathered}
\]
Put $k=(p-l)/(1-l)$. There is a computable $\epsilon_0>0$ such that every $0<\epsilon\le\epsilon_0$ admits strictly positive compatible tables $p(\epsilon,t)$, $|t|\le\epsilon/2$, with
\begin{align*}
 (L_0,U_0)&=(1/2,1/2+k),\\
 (L_1,U_1)&=(1/2-k+\epsilon,1/2+\epsilon),\\
 I(p(\epsilon,t))&=[-2k+\epsilon,t].
\end{align*}
At fixed $\epsilon$, identical summaries conceal strict negativity at $t=-\epsilon/2$ and admission of both signs at $t=\epsilon/2$. This distinction holds on paired observation neighborhoods and on an open neighborhood of each such design in the full six-dimensional strict-design space.
\end{theorem}

The range for $\rho=(p-l)/(u-l)$ is sufficient for strictly positive witness slacks. The class includes common-$\Gamma$ JSM designs with $\pi=(p,1-p)$, $p\in[1/5,27/100]$ and $\Gamma\in[9/5,11/5]$. On every fixed-$\epsilon$ family, any decoder of the four endpoints has worst-case upper-endpoint error at least $\epsilon/2$. The proof below gives the tables, witnesses and positive-radius boundary limits at fixed interior $l,u$.

\paragraph{Design class.}
Let
\begin{equation}\label{eq:mirror-design-class}
 0<l<\frac15,\qquad l<u<\frac9{20},\qquad
 \frac13<\rho:=\frac{p-l}{u-l}<\frac49.
\end{equation}
Use primitive envelopes and prescribed propensities
\[
 (\ell_0,u_0,\pi_0)=(l,u,p),\qquad
 (\ell_1,u_1,\pi_1)=(1-u,1-l,1-p).
\]
This is a three-parameter class of strict nonaligned designs. It
contains the manuscript's design $(l,u,p)=(1/7,2/5,1/4)$, for which
$\rho=5/12$. It also intersects the symmetric JSM in an open set of
its two free parameters $(p,\Gamma)$; the inequalities above, rather
than a particular value of $\Gamma$, determine that intersection.
Write
\[
 D=u-l,\qquad k=\frac{p-l}{1-l}>0,\qquad \sigma=1-u-l>0.
\]

\paragraph{Explicit observation family.}
For every design satisfying \eqref{eq:mirror-design-class}, there is
a computable $\epsilon_0>0$ such that, for all
$0<\epsilon\leq\epsilon_0$ and $|t|\leq\epsilon/2$, the tables
\begin{align*}
 p_0(\epsilon,t)&=
 \begin{pmatrix}1-p-a&a\\p-b-t&b+t\end{pmatrix},&
 p_1(\epsilon)&=
 \begin{pmatrix}p-b&b\\1-p-c&c\end{pmatrix},\\
 a&=\frac{1-l}{2},&
 b&=p-\frac{u+l}{4}-\frac{1-l}{2}\epsilon,\\
 c&=\frac{1+l}{2}-p+(1-l)\epsilon
 \end{align*}
are strictly positive and compatible at their prescribed propensities.
Their four full-model sharp mean endpoints are
\[
 (L_0,U_0,L_1,U_1)
 =\left(\frac12,\frac12+k,\frac12-k+\epsilon,
                  \frac12+\epsilon\right),
\]
whereas the complete sharp ATE interval is
\[
 I(p_0(\epsilon,t),p_1(\epsilon))
   =[-2k+\epsilon,t].
\]
At each fixed $\epsilon$, the four endpoints are identical across
$t$. Thus $t=-\epsilon/2$ identifies a strictly negative ATE and
$t=\epsilon/2$ admits both signs. The two conclusions persist on
paired relatively open four-dimensional observation neighborhoods.
Moreover this same-summary decision failure holds for every design
in an open neighborhood, in the full six-dimensional space of
strict two-environment designs, of each design in
\eqref{eq:mirror-design-class}.

\begin{corollary}[An explicit symmetric-JSM parameter rectangle]
\label{cor:mirror-jsm-rectangle}
The theorem applies to every common-$\Gamma$ symmetric JSM design
with propensities $(p,1-p)$ satisfying
\[
 \frac15\leq p\leq\frac{27}{100},\qquad
 \frac95\leq\Gamma\leq\frac{11}{5}.
\]
Thus same-summary sign ambiguity is proved on an entire rectangle
of sensitivity parameters and prescribed treatment propensities.
\end{corollary}

\begin{proof}
For a symmetric JSM,
$l=p/\{\Gamma-(\Gamma-1)p\}$ and
$u=\Gamma p/\{1+(\Gamma-1)p\}$, while the complementary
propensity has endpoints $(1-u,1-l)$. Direct cancellation gives
\[
 \rho=\frac{p-l}{u-l}
      =\frac{1+(\Gamma-1)p}{\Gamma+1}.
\]
On the stated rectangle, monotonicity and endpoint substitution give
\[
 l\leq\frac{15}{88}<\frac15,\qquad
 u\leq\frac{297}{662}<\frac9{20},\qquad
 \frac{31}{80}\leq\rho\leq\frac{76}{175},
\]
and $1/3<31/80<76/175<4/9$. All primitive conditions hold.
\end{proof}

\subsection{Proof: local bounds and three attaining mean witnesses}

For arbitrary nearby observations write
$y=(a_0,b_0,a_1,b_1)$. The following coordinate bounds are valid
in the full model, by the indicated single-environment score limits:
\begin{equation}\label{eq:mirror-local-four}
 L(y)=\frac{a_0}{1-l}\leq m_0\leq L(y)+k,
 \qquad
 R(y)=\frac{b_1}{1-l}\leq m_1\leq R(y)+k.
\end{equation}
The target supplies the two treatment-zero bounds, and the source
supplies the two treatment-one bounds. These are not asserted sharp
until the witnesses below have been checked.

The cross-fiber comparison $A_0\leq C_1$, multiplied by $D$, is
\begin{equation}\label{eq:mirror-oblique}
 -\sigma m_0+m_1\leq d(y),\qquad
 d(y)=u-p-a_0+b_0+a_1+b_1.
\end{equation}
For the displayed tables $L(y)=1/2$, $R(y)=1/2-k+\epsilon$ and
$d(y)=(l+u)/2+t$. Combining \eqref{eq:mirror-oblique} with
$m_0\geq1/2$ gives
\[
 m_1-m_0\leq d(y)-(1-\sigma)m_0\leq t.
\]
The coordinate bounds give the lower ATE bound $-2k+\epsilon$.
The required three mean witnesses are
\begin{equation}\label{eq:mirror-three-witnesses}
 v^\tau=(1/2,1/2+t),\qquad
 v^H=(1/2+(\epsilon-t)/\sigma,1/2+\epsilon),\qquad
 v^-=(1/2+k,1/2-k+\epsilon).
\end{equation}
The first attains the lower $m_0$ and the upper ATE. The second
attains the upper $m_1$. The third simultaneously attains the upper
$m_0$, lower $m_1$, and lower ATE. Proving their feasibility proves
all six extrema over the same full model.

\subsection{The two base points and their strict slacks}

At $\epsilon=t=0$, the observation coordinates are
\[
 y_*=( (1-l)/2,z,z,(1+l)/2-p),\qquad z=p-(u+l)/4.
\]
The first two witnesses coincide at $V=(1/2,1/2)$, while the third
is $W=(1/2+k,1/2-k)$.

For $V$, take the strictly positive common law $q=(1,1,1,1)/4$.
The two score vectors are
\[
 h_0=(u,l+(4\rho-1)D,l,l),\qquad
 h_1=(1-l,1-u,1-l,1-u+(2-4\rho)D).
\]
Their means are $p,1-p$ and they reproduce $y_*$. The two
nonendpoint entries are strictly inside their envelopes because
$1/4<\rho<1/2$. In the exact local-rectangle and cross-fiber
half-plane representation, the only tight rows at $V$ are the
target lower $m_0$, the source upper $m_1$, and $A_0\leq C_1$.
To see uniqueness of the tight comparison, the only tight lower
fiber piece is $A_0=q_{11}$ (target low $00$ and high $11$ are
zero), and the only tight upper piece is $C_1=q_{11}$ (source
low $10$ and high $01$ are zero). Every other relevant group cell
in each expression has positive mass. The common Fr\'echet pieces
are strict since all four type masses are $1/4$. The local group
marginals have only the two zero complements corresponding to the
stated coordinate constraints.

For completeness, the following table gives all sixteen local
rectangle slacks at $W$, multiplied by their positive denominators.
The order in each block is lower $m_0$ twice, upper $m_0$ twice,
lower $m_1$ twice, upper $m_1$ twice; the same order is used by
\texttt{exact\_rows} in the certificate program.
\[
\begin{array}{c|rrrr}
 &\multicolumn{4}{c}{\text{treatment-zero constraints}}\\
 e=0 &(1-l)k&D(1/2-k)&D/2-(1-u)k&0\\
 e=1 &3D/4-\sigma k&D/4+lk&(1-2l)k-D/4&D/4-uk\\ \hline
 &\multicolumn{4}{c}{\text{treatment-one constraints}}\\
 e=0 &3D/4-(1+D)k&D/4-lk&k-D/4&D/4+uk\\
 e=1 &0&D/2-(1-u)k&D(1/2-k)&(1-l)k
\end{array}
\]
The sixteen cross-fiber slacks, all multiplied by $D$, are
\[
\begin{array}{c|rrrr}
 (e,f)& C_f-A_e&H_f-A_e&C_f-B_e&H_f-B_e\\ \hline
 (0,0)&(1+u)k&D/2+(1-u)k&D/2-(1-u)k&D-(1+u)k\\
 (0,1)&(1+\sigma)k&D/2+(l+u)k&D/2-(l+u)k&D-(1+\sigma)k\\
 (1,0)&D/2+(l+u)k&D(1-k)&Dk&D/2-(l+u)k\\
 (1,1)&D/2+(1-u)k&D-(1-u-2l)k&(1-u-2l)k&D/2-(1-u)k
\end{array}
\]
Every entry other than the two displayed local zeros is strictly
positive on \eqref{eq:mirror-design-class}. Here is a sufficient
set of bounds making this a direct algebraic check:
\begin{align*}
 &\frac{D}{4(1-2l)}<k<\frac{D}{4u},\qquad
 k<\frac{3D}{4(1+D)},\qquad k<\frac{D}{2(1-u)},\\
 &k<\frac{D}{2(l+u)},\qquad
 k<\frac{D}{1+u},\qquad k<\frac{D}{1+\sigma},\qquad
 0<k<1/2,\qquad 1-u-2l>0.
\end{align*}
For example, the first lower bound follows from
$(1-l)/\{4(1-2l)\}<1/3<\rho$, using $l<1/5$.
The first upper bound follows from
$(1-l)/(4u)>4/9>\rho$, using $l<1/5,u<9/20$.
The remaining upper bounds follow by substituting
$k=D\rho/(1-l)$; their design thresholds exceed $4/9$ under the
same restrictions. Also $k>D/4$, $k<D/(4l)$,
$\sigma k<4D/9<3D/4$, and
$(1-u-2l)k<D$, covering the remaining table entries.
Thus $W$ is feasible and its only tight rows are the target upper
$m_0$ and source lower $m_1$. Exact fiber reconstruction supplies
an original-mass causal witness; no geometric outer relaxation is
being substituted for the full model.

\paragraph{Why the $\rho$ range is sufficient, not necessary.}
The interval in \eqref{eq:mirror-design-class} makes the displayed
slacks positive uniformly over the stated $l,u$ ranges. It does not
characterize every mirror design admitting sign separation.
For fixed $0<l<1/5$ and $l<u<9/20$, those strict slacks remain
positive also at $\rho=1/3$ and $4/9$. In particular the two potentially
limiting local slacks are
\[
 (1-2l)k-D/4=\frac{D(1-5l)}{12(1-l)}\quad(\rho=1/3),\qquad
 D/4-uk=\frac{D\{9(1-l)-16u\}}{36(1-l)}\quad(\rho=4/9),
\]
and both are positive; all other inequalities above retain their
strict margins. The base cells are positive and the active-row
perturbations below remain nonnegative. Continuity therefore gives
$\liminf\epsilon_0>0$ as $\rho$ alone approaches either endpoint.
This is not uniform in $l,u$: the displayed slacks can vanish when
those parameters simultaneously approach their outer boundaries.

\subsection{An explicit arithmetic radius}

Represent the full mean region by its thirty-two half-planes
\[
 n_j^\top m\leq h_j(y),\qquad j=1,\ldots,32,
\]
consisting of all eight local rectangle rows per environment and
the four comparison types for every ordered environment pair.
This representation is exact by the manuscript's local fiber lemma.
Both $h_j$ and all three witnesses in
\eqref{eq:mirror-three-witnesses} are affine in $(\epsilon,t)$.
For witness $v$, write its slack as
\[
 s_{vj}+c_{vj}\epsilon+e_{vj}t.
\]
All $s_{vj}$ are positive except the active rows just described.
For those zero rows, their perturbation slacks are respectively:
\[
\begin{array}{c|ccc}
 v^\tau&0&\epsilon-t&0\\
 v^H &(\epsilon-t)/\sigma&0&0
\end{array}
\]
for target lower $m_0$, source upper $m_1$, and $A_0\leq C_1$;
the last row is scaled if its normal is not normalized by $D$.
Both active rows for $v^-$ remain identically zero. Therefore the
zero-slack rows stay feasible whenever $|t|\leq\epsilon/2$.

Define $\epsilon_0$ as the minimum of $k/2$, all positive ratios
\[
 \frac{s_{vj}}{2(|c_{vj}|+|e_{vj}|/2)}
\]
with nonzero denominator, and the same ratios for the eight observed
cell probabilities. Their base probabilities are strictly positive.
Empty or zero-denominator ratios are omitted. This is a finite,
explicitly computable positive number. It ensures every previously
strict slack and every observed cell remains positive, while the
active slacks remain nonnegative. Thus all three witnesses are
feasible. It also ensures $-2k+\epsilon<0$, so the two sign
conclusions in the theorem follow. No LP or asymptotic limiting
argument is required to compute the radius.

\subsection{Paired observation neighborhoods and open design neighborhoods}

Fix a proved $\epsilon>0$ and let $t_\pm=\pm\epsilon/2$.
At each corresponding observation collection the active mean
witnesses can instead be defined for arbitrary $y$ by
\[
 v^\tau(y)=(L(y),d(y)+\sigma L(y)),\qquad
 v^H(y)=((R(y)+k-d(y))/\sigma,R(y)+k),\qquad
 v^-(y)=(L(y)+k,R(y)).
\]
These maps preserve their two active rows exactly. All other rows
are strict, as proved by the radius construction, so feasibility
persists on some full four-dimensional observation box. Consequently
throughout these boxes,
\begin{align*}
 (L_0,U_0,L_1,U_1)&=(L(y),L(y)+k,R(y),R(y)+k),\\
 U_\tau(y)&=u-p-\frac{1+u}{1-l}a_0+b_0+a_1+b_1,\\
 L_\tau(y)&=R(y)-L(y)-k.
\end{align*}
Adding the same perturbation $h$ to the two observations keeps
their four mean endpoints identical. The two upper endpoints differ
by $\epsilon$, and their signs remain opposite when
\[
 \|h\|_\infty<\frac{\epsilon}
 {4\{(1+u)/(1-l)+3\}},
\]
as well as satisfying the feasible-box radius. The lower endpoint
remains negative on a sufficiently small common box.

On each fixed-$\epsilon$ family, every deterministic decoder of
the design and four endpoints has worst-case upper-endpoint error
at least $\epsilon/2$, attained on that family by returning zero.
On its certified observation chamber, retaining $w=b_0+a_1$
additionally recovers the upper endpoint through
\[
 U_\tau=u-p-(1+u)L_0+(1-l)L_1+w.
\]
This design-parametric repair recovers the upper endpoint on the
certified observation chamber.

Finally consider a full design vector
$\theta=(\ell_0,u_0,\pi_0,\ell_1,u_1,\pi_1)$ near a proved mirror
design. The active coordinate rows remain
\[
 L_0=a_0/(1-\ell_0),\quad
 U_0=1-(1-\pi_0-a_0)/(1-\ell_0),\quad
 L_1=b_1/u_1,\quad U_1=1-(\pi_1-b_1)/u_1.
\]
The retained cross row remains $-A m_0+B m_1\leq d_\theta(y)$
with $0<A<B$, by continuity from $A/B=\sigma\in(0,1)$.
Its right-hand side has nonzero $b_0$ derivative $1/(u_0-\ell_0)$.
The three two-row active intersections depend continuously on
$(\theta,y)$ and every other row is strict. Hold $a_0,a_1,b_1$
at their constructed values, and solve the affine equation
\[
 d_\theta(y)/B+(A/B-1)L_0(y)=0
\]
for $b_0=b_0^*(\theta)$. This solution is continuous and equals
the nominal center value. Varying $b_0$ slightly above and below
$b_0^*(\theta)$ leaves all four mean endpoints fixed and gives
opposite signs of $U_\tau$, while the attaining intersections
remain feasible. Positivity and a negative lower endpoint persist
as well. This proves the full-dimensional open design extension.
The extension is local to the proved mirror-design class.

\paragraph{Verification.}
The independently assembled program
\path{research19/code/design_decodability.py} with option \texttt{--verify-mirror}
checks the thirty-two exact half-planes, computes the stated radius,
reconstructs original-mass witnesses, and reconstructs the entire
mean polygon with Fraction arithmetic. Twenty-three rational designs,
including the four corners of the symmetric-JSM parameter rectangle,
give forty-six exact full-polygon checks and 138 causal-witness
checks. These checks validate formulas and implementation; the
continuum theorem is established by the positive-slack algebra above.
The record is
\path{research19/results/mirror_design_certificates.json}.

The design-class result covers the mirror class and its local open design neighborhoods at zero drift. Proposition~\ref{thm:insufficiency} supplies an exact positive-drift family using separate observed tables.

\section{Identical marginal summaries can conceal different causal conclusions}
\label{app:decision-information}

The common subtraction interval is $[-109/480-\delta,11/480+\delta]$. Identical inputs at $t=\pm1$ force every four-endpoint decoder to incur worst-case upper error at least $1/160$. The formula is proved below by full-mass dual identities and attaining causal models. This family also gives the exact negativity frontier $\delta<-t/160$, a model-specific breakdown frontier \cite{breakdown}.

This appendix proves Proposition~\ref{thm:insufficiency}, including positive drift and an explicit neighborhood of paired observations. The target is environment $0$. Selection bounds hold within each environment's own response-type law: $\ell_eq_{e,ij}\le z_{e,ij}\le u_eq_{e,ij}$. Fix
\[
 (\pi_0,\pi_1)=(1/4,3/4),\qquad
 (\ell_0,u_0)=(1/7,2/5),\quad(\ell_1,u_1)=(3/5,6/7),
\]
and $\TV(q_1,q_0)\le\delta$. These are the pair-selection JSM envelopes for $\Gamma_0=\Gamma_1=2$. No selection vector or response-type law is fixed beyond these restrictions.

For $t\in[-1,1]$, the observations are
\[
 p_0(t)=\frac1{160}\begin{pmatrix}60&60\\25-t&15+t\end{pmatrix},\qquad
 p_1=\frac1{160}\begin{pmatrix}27&13\\74&46\end{pmatrix}.
\]
Rows are treatment and columns outcome. All cells are positive, with minimum $13/160$, and the propensities remain fixed. Write $y=(a_0,b_0,a_1,b_1)$, where $a_e=p_e(0,1)$ and $b_e=p_e(1,1)$. At the family center,
\[
 y_*=(3/8,3/32,13/160,23/80),\qquad
 y(t)=y_*+(0,t/160,0,0).
\]
Define $L=7a_0/6$, $R=7b_1/6$, and
\[
 V(y)=3/20-49a_0/30+b_0+a_1+b_1.
\]
Then $V(y(t))=t/160$.

\subsection{Five endpoint inequalities in the original mass variables}

Let $\mathcal L_{e,ij}=z_{e,ij}-\ell_eq_{e,ij}$ and $\mathcal U_{e,ij}=u_eq_{e,ij}-z_{e,ij}$ denote nonnegative selection slacks. Let $m_a=\Pr_0(Y(a)=1)$ and $n_a=\Pr_1(Y(a)=1)$. Observation consistency and normalization give
\begin{align*}
 m_0-L&=\tfrac76(\mathcal L_{0,10}+\mathcal L_{0,11}),\\
 \tfrac18+L-m_0&=\tfrac76(\mathcal L_{0,00}+\mathcal L_{0,01}),\\
 n_1-R&=\tfrac76(\mathcal U_{1,01}+\mathcal U_{1,11}),\\
 \tfrac18+R-n_1&=\tfrac76(\mathcal U_{1,00}+\mathcal U_{1,10}).
\end{align*}
Since $|n_1-m_1|\le\delta$, these identities imply the four bounds
\[
 L\le m_0\le\tfrac18+L,\qquad R-\delta\le m_1\le\tfrac18+R+\delta.
\]
For the joint endpoint, put $\varphi=(1,2/5,6/7,0)$ in response-type order $00,01,10,11$. Its component range is one. Direct expansion in the original sixteen mass variables gives
\begin{align*}
 V(y)+\delta-(m_1-m_0)
 ={}&\tfrac{19}{30}(\mathcal L_{0,10}+\mathcal L_{0,11})
       +\mathcal U_{0,00}+\mathcal L_{0,11}\\
    &+\mathcal L_{1,01}+\mathcal U_{1,10}
       +\delta-\varphi^\top(q_1-q_0).
\end{align*}
The last term is nonnegative because $\sum(q_1-q_0)=0$ and
$\varphi^\top(q_1-q_0)\le\TV(q_1,q_0)\le\delta$. Thus $m_1-m_0\le V(y)+\delta$. This dual identity certifies the full model's upper bound for arbitrary observations.

\subsection{Explicit attaining models on the complete parameter rectangle}

For any fiber triple $(r_0,r_1,c)$ define
\[
 q(r_0,r_1,c)=(1-r_0-r_1+c,r_1-c,r_0-c,c).
\]
The following table gives a target triple $x_0$ and a source triple $x_1$ attaining each endpoint. Set
\[
 C=\frac{35}{9}\{a_0-b_0-3/20-3L/5+2(L+V(y)+\delta)/5\}.
\]
\begin{center}
\begin{tabular}{lcc}
\toprule
Endpoint & Target triple $x_0$ & Source triple $x_1$\\
\midrule
$L_0$ & $(L,3/8,3/20)$ & $(L,3/8,3/20)$\\
$U_0$ & $(1/8+L,3/8,7/30)$ & $(1/8+L,3/8,7/30)$\\
$L_1$ & $(11/24,R-\delta,1/6-\delta)$ & $(11/24-\delta,R,1/6-\delta)$\\
$U_1$ & $(25/48+\delta,1/8+R+\delta,1/16-2\delta)$ & $(25/48+2\delta,1/8+R,1/16-2\delta)$\\
$U_\tau$ & $(L,L+V(y)+\delta,C)$ & $(L-\delta,L+V(y),C-\delta)$\\
\bottomrule
\end{tabular}
\end{center}
The two triples in the first two rows give the same type law. In the last three rows, respectively,
\[
 q_1-q_0=(0,\delta,-\delta,0),\quad
 (0,-\delta,\delta,0),\quad(\delta,0,0,-\delta).
\]
Hence the drift is exactly $\delta$ in those rows.

Selection coordinates can be reconstructed explicitly. For environment $e$ with triple $(r_0,r_1,c)$, let $k_e=z_{e,11}$. The observation equations give
\[
 z_e=(\pi_e-b_e-r_0+a_e+k_e,\ b_e-k_e,\ r_0-a_e-k_e,\ k_e).
\]
Feasibility of $q_e$ and a selection vector is equivalent to type
nonnegativity and $c$ belonging to that environment's exact fiber
from Appendix~\ref{sec:aligned-envelope}, using its own envelope.
The same endpoint-score group construction supplies $k_e$ and $z_e$;
no common law across these two local reconstructions is required.
For every table row, direct rational substitution at the four parameter vertices
\[
 (t,\delta)\in\{-1,1\}\times\{0,1/100\}
\]
gives nonnegative $q_e$ and all local lower-fiber pieces no greater than $c$, and all upper pieces no smaller than $c$. These comparisons and the type masses are affine in $(t,\delta)$, so they remain valid throughout the rectangle. The reconstructed selection vector satisfies all sixteen envelope inequalities and every observed cell. Each objective attains its dual bound, proving
\[
 (L_0,U_0,L_1,U_1)=(7/16,9/16,161/480-\delta,221/480+\delta),
 \qquad U_\tau=t/160+\delta
\]
for all $t\in[-1,1]$ and $\delta\in[0,1/100]$. These causal witnesses establish compatibility throughout the parameter rectangle.

\subsection{An explicit neighborhood with indistinguishable summaries}

The endpoint expressions and witnesses above are affine in the full four-dimensional observation vector. The reduced system has forty-eight inequalities: sixteen local fiber inequalities in each environment, including type nonnegativity, and sixteen signed TV inequalities. For a witness row, write a reduced slack as $s_i(t,\delta)+g_i^\top h$, where the same perturbation $h=y-y(t)$ is added to the observations. Whenever a base slack is zero its observation gradient is zero. Exact substitution at $t=\pm1$ and $\delta\in\{0,1/320\}$ gives the following minimum ratios over rows with nonzero gradients:
\begin{center}
\begin{tabular}{lc}
\toprule
Witness & $\min_i s_i/\|g_i\|_1$\\
\midrule
$L_0$ & $19/1820$\\
$U_0$ & $3/3920$\\
$L_1$ & $19/5460$\\
$U_1$ & $17/10920$\\
$U_\tau$ & $1/158$\\
\bottomrule
\end{tabular}
\end{center}
Their common minimum $3/3920$ exceeds $1/2000$. Affineness in $(t,\delta)$ therefore gives nonnegative slacks whenever
\[
 -1\le t\le1,\qquad0\le\delta\le1/320,\qquad\|h\|_\infty\le1/2000.
\]
Observed cells remain positive. Thus all five endpoint formulas hold on the entire perturbed parameter domain. The independent verifier additionally checks all 320 observation/parameter corners directly in the original mass model; convexity establishes their whole domain.

For a common perturbation $h$, compare $y_-=y(-1)+h$ with $y_+=y(1)+h$ at the same drift budget. Their sharp mean endpoints coincide because those expressions depend only on $a_0,b_1,\delta$. Write
\[
 v(h)=-49h_{a_0}/30+h_{b_0}+h_{a_1}+h_{b_1},\qquad
 |v(h)|\le139/60000.
\]
Their sharp joint uppers obey
\[
 U_\tau(y_-,\delta)=-1/160+\delta+v(h)\le-97/120000<0,
\]
\[
 U_\tau(y_+,\delta)=1/160+\delta+v(h)\ge59/15000>0.
\]
The $U_0$ witness has negative ATE $3/8-(1/8+7a_0/6)$ throughout this domain, so the positive branch admits both effect signs. The common marginal subtraction interval also contains zero: its lower and upper are
\[
 -109/480-\delta+\tfrac76(h_{b_1}-h_{a_0}),\qquad
 11/480+\delta+\tfrac76(h_{b_1}-h_{a_0}).
\]
Both signs are strict. The interiors of the two paired boxes are relatively open at the fixed propensities. Both strict sign conclusions therefore persist over paired observation neighborhoods and a nonzero interval of positive drift budgets.

\subsection{Decision impossibility, a sharp summary error, and a local repair}

Let $S(p,\delta)=(L_0,U_0,L_1,U_1)$ denote the four full-model sharp mean endpoints. At any fixed $\delta<1/160$, the displayed $t=\pm1$ laws have identical $S$, but only the negative branch has $U_\tau<0$. Consequently no function of the design, drift budget, and $S$ can always decide whether the joint model identifies a negative ATE. A procedure that declares negativity only when it is justified for every compatible law giving that summary must withhold that declaration even at the negative branch. This establishes a population information restriction on the four-endpoint summary.

There is also a quantitative endpoint consequence. For any deterministic summary decoder $g$ and any fixed $\delta\in[0,1/100]$,
\[
 \sup_{-1\le t\le1}|g(S(p(t),\delta))-(t/160+\delta)|\ge1/160.
\]
The two extreme upper endpoints are $1/80$ apart and the decoder returns one value. On this specified family the bound is sharp: returning $\delta$ attains worst-case error $1/160$. The attained worst-case error here is specific to the displayed family.

On the certified observation domain, the discarded dependence restriction yields
\begin{equation}\label{eq:summary-local}
 U_\tau=3/20-49a_0/30+b_0+a_1+b_1+\delta.
\end{equation}
For the displayed $t$ family the exact negativity condition is
\begin{equation}\label{eq:sign}
 \delta<-t/160.
\end{equation}
The local formulas also identify a sufficient repair. Retain the single observable scalar $w=b_0+a_1$ in addition to the four endpoints. On the certified domain,
\[
 U_\tau=\frac3{70}-\frac75L_0+\frac67U_1+w+\frac\delta7.
\]
Thus the lost scalar recovers the exact upper endpoint and its sign decision in this regime. The formula gives a local repair on this certified chamber. For general designs, the fiber representation supplies the constraints needed for joint optimization.

The exact producer and certificate are \path{research17/code/decision_information.py} and \path{research17/results/decision_information_certificate.json}. The separate verifier \path{research17/code/independent_decision_verifier.py} reconstructs the primitive mass restrictions, the original-mass dual identities, and all witness checks with standard-library Fraction arithmetic, without importing the producer.

\section{Outcome-specific restrictions and mismatch loss}
\label{app:uniform-mismatch}
The first comparison below completes Theorem~\ref{thm:universal}: restrictions conditional on one potential outcome preserve the rectangle. The remaining subsections quantify the missing corner when restrictions instead act on the response pair. All regions and widths are from the full common-law model.

\subsection{Why outcome-specific restrictions preserve rectangularity}
Replace the pair-type limits by mass inequalities
\[
 \ell_e\Pr(Y(a)=y)\le\Pr_e(A=1,Y(a)=y)\le u_e\Pr(Y(a)=y)
 \quad(a,y\in\{0,1\}).
\]
For any unrestricted common-law witness, its selected marginal masses are
\[
 \Pr_e(A=1,Y(0)=1)=m_0-a_e,\quad
 \Pr_e(A=1,Y(0)=0)=\pi_e-m_0+a_e,
\]
\[
 \Pr_e(A=1,Y(1)=1)=b_e,\quad
 \Pr_e(A=1,Y(1)=0)=\pi_e-b_e.
\]
The weaker restrictions are therefore equivalent to
\begin{align*}
 \ell_em_0&\le m_0-a_e\le u_em_0,&
 \ell_e(1-m_0)&\le\pi_e-m_0+a_e\le u_e(1-m_0),\\
 \ell_em_1&\le b_e\le u_em_1,&
 \ell_e(1-m_1)&\le\pi_e-b_e\le u_e(1-m_1).
\end{align*}
They describe the same local coordinate rectangle $R_e$ as a single pair-envelope model, but their truth is independent of the joint dependence coordinate and of the selected coupling.

For completeness, unrestricted selection $0\le z_e\le q$ has local mean rectangle
$[a_e,a_e+\pi_e]\times[b_e,b_e+1-\pi_e]$. Splitting $q$ into treated and untreated measures gives its exact fiber
\[
 J_e^{\rm IV}(m)=
 [\max\{0,m_0+m_1-1,m_1+r_e-1,m_0-r_e\},
 \min\{m_0,m_1,s_e,m_0+m_1-s_e\}].
\]
Cross comparisons of the two noncommon lower and upper pieces reduce to
\[
 m_1\le1+s_f-r_e,\quad m_0\ge r_e+s_f-1,\quad
 m_0\le r_e+s_f,\quad m_1\ge s_f-r_e.
\]
Thus the unrestricted region $P_{\rm IV}$ is a rectangle, with corners attainable by common fiber intersection, recovering the binary IV property \cite{rr}. For every $m$ in it, any unrestricted witness obeys the weaker limits if and only if $m\in\bigcap_eR_e$. Hence
\[
 P_{\rm weak}=P_{\rm IV}\cap\bigcap_eR_e,
\]
which is rectangular even when the envelopes differ. This also extends to coordinate- and outcome-specific marginal limits: every extra inequality still restricts only one mean. The claim requires zero drift and no additional joint-law or cross-selection restrictions.

\subsection{An exact nested-model comparison with binary IV}
The unrestricted mass model is exactly realizable as a randomized-IV
model with exclusion. Draw $T\sim q$, independent uniforms $V_e$, and an
independent environment label $E$. Set
$A(e)=\mathbf1\{V_e\le z_e(T)/q(T)\}$ on positive-mass types,
$A=A(E)$ and $Y=T_A$. Then $E$ is independent of the entire potential
assignment and outcome vector. Conversely, consistency and randomized
IV assignment produce unrestricted masses $0\le z_e\le q$.
Finite type-specific score bounds therefore define a nested submodel,
not a model outside the IV framework. No monotonicity is imposed.

Use Proposition~\ref{thm:insufficiency}'s positive tables at $t=0$:
\[
p_0=\frac1{160}\begin{pmatrix}60&60\\25&15\end{pmatrix},\qquad
p_1=\frac1{160}\begin{pmatrix}27&13\\74&46\end{pmatrix}.
\]
Substitution into the unrestricted rectangles and cross comparisons
above gives
\[
P_{\rm IV}=[3/8,5/8]\times[23/80,43/80].
\]
For $\Gamma=2$, the finite envelopes are $(1/7,2/5)$ and $(3/5,6/7)$.
Substitution into the finite fiber comparisons gives exactly
\[
\begin{split}
P_{\rm pair}=\{m:\;&7/16\le m_0\le9/16,\quad
161/480\le m_1\le221/480,\\
&-256m_0+560m_1\le133,\quad70m_0+32m_1\le53\}.
\end{split}
\]
The six vertices, in boundary order, are
\[
(7/16,161/480),\ (9/16,161/480),\ (9/16,109/256),\quad
(41/75,221/480),\ (749/1536,221/480),\ (7/16,7/16).
\]
The first oblique inequality and $m_0\ge7/16$ give
$m_1-m_0\le0$, attained at the last vertex. The sharp marginal upper
difference is instead $221/480-7/16=11/480$.
At that missing corner, the local dependence fibers are
\[
J_0=[889/4320,1519/4320],\qquad J_1=[7/96,7/60],
\]
which are disjoint. The weaker outcome-specific model retains the whole
box $[7/16,9/16]\times[161/480,221/480]$. This isolates the effect of
pair-type restrictions using the same observed law and treatment
categories. This same-table comparison locates the oblique constraints introduced
by pair-type selection limits within the IV model \cite{song}. The standalone rational
elimination and original-mass reconstruction are in
\path{research22/literature/verify_nested_iv_example.py}.

\subsection{Mismatch size and the uniform error bound}
At zero drift, information can be lost even when the mismatch is small, but its size can be controlled. Define a dimensionless envelope mismatch
\begin{equation}\label{eq:mismatch-size}
 \chi=\max_{e,f}\max\left\{
 \frac{|\alpha_e-\alpha_f|}{\beta_e+\beta_f-1},\
 \frac{|\beta_e-\beta_f|}{\alpha_e+\alpha_f+1}\right\}.
\end{equation}
The denominators are positive for strict envelopes, and $\chi=0$ exactly at alignment. Let $W_a=U_a-L_a$ be the \emph{full-model} sharp marginal widths.

\begin{theorem}[Uniform mismatch error]\label{thm:gap-bound}
For every compatible binary common-law model with strict envelopes,
\begin{equation}\label{eq:uniform-gap}
 0\leq G(p)\leq\frac{\chi}{1+\chi}(W_0+W_1)
 \leq\frac{2\chi}{1+\chi}.
\end{equation}
The same bound holds for the lower composition gap. For every $0<\chi\le1/29$, strictly positive compatible three-environment tables attain the first upper bound. Its coefficient cannot be reduced at those mismatch values, for either endpoint.
\end{theorem}

\subsection{All boundary directions of the mean region}
Let $\alpha_e=(1-u_e)/(u_e-\ell_e)$ and $\beta_e=u_e/(u_e-\ell_e)$. Then $\alpha_e>0$ and $\beta_e>1$. In the notation of \eqref{eq:local-fiber}, put
\[
 k_e^A=(r_e-u_e)/D_e,\quad k_e^B=(-r_e+\ell_e)/D_e,
 \quad k_e^C=s_e/D_e,\quad k_e^H=-s_e/D_e.
\]
The exact fiber pieces are
\begin{align*}
 A_e&=k_e^A-\alpha_em_0+\beta_em_1,&
 B_e&=k_e^B+(\alpha_e+1)m_0-(\beta_e-1)m_1,\\
 C_e&=k_e^C-\alpha_em_0-(\beta_e-1)m_1,&
 H_e&=k_e^H+(\alpha_e+1)m_0+\beta_em_1.
\end{align*}
As proved in Appendix~\ref{sec:aligned-envelope}, on the intersection of the local mean rectangles the common fiber is nonempty exactly when all ordered pairs satisfy the four comparisons below. The derivation of the local fiber uses each environment's own envelope and does not require alignment.
\[
\begin{array}{c|c|c}
\text{Comparison}&\text{Left-hand side}&\text{Right-hand side}\\ \hline
A_e\le C_f&-(\alpha_e-\alpha_f)m_0+(\beta_e+\beta_f-1)m_1&k_f^C-k_e^A\\
A_e\le H_f&-(\alpha_e+\alpha_f+1)m_0+(\beta_e-\beta_f)m_1&k_f^H-k_e^A\\
B_e\le C_f&(\alpha_e+\alpha_f+1)m_0-(\beta_e-\beta_f)m_1&k_f^C-k_e^B\\
B_e\le H_f&(\alpha_e-\alpha_f)m_0-(\beta_e+\beta_f-1)m_1&k_f^H-k_e^B
\end{array}
\]
Consequently $P$ is a compact convex polygon, possibly a segment or singleton. Its nonaxis boundary normals belong to two families:
\[
 |n_0|/|n_1|\le\chi\quad\hbox{or}\quad |n_1|/|n_0|\le\chi,
\]
with the respective denominator nonzero. All defining directions are covered; this is not a restriction to the single active inequality in the converse construction.

\subsection{Proof of the uniform bound}
Transform to deficit coordinates $x=m_0-L_0$, $y=U_1-m_1$ and let $Q$ be the transformed region. Then
\[
 Q\subseteq[0,W_0]\times[0,W_1],\qquad \min_Qx=\min_Qy=0,
 \qquad G=\min_Q(x+y).
\]
Using a point with $x=0$ and a point with $y=0$ gives
\begin{equation}\label{eq:elementary-width}
 G\le\min(W_0,W_1)\le(W_0+W_1)/2.
\end{equation}
Thus the desired inequality is immediate when $\chi\ge1$ or one width is zero. If $\chi=0$, all $\alpha_e$ and $\beta_e$ agree; all cross comparisons are coordinate restrictions and $P$ is a rectangle. Hence $G=0$.

Suppose $0<\chi<1$ and first that $Q$ has two-dimensional interior. Its lower boundary
\[
 f(x)=\min\{y:(x,y)\in Q\},\qquad 0\le x\le W_0,
\]
is continuous, convex, and piecewise affine. Put $h=f(0)\le W_1$ and $t=\min\{x:f(x)=0\}\le W_0$. If $h=0$ there is no gap. Otherwise $t>0$ and $f$ decreases on $[0,t]$. Every edge comes from a defining half-plane. Writing a decreasing edge slope as $-s$, the normal restrictions give either $s\le\chi$ or $s\ge1/\chi$. Convexity orders the steep edges before the shallow edges. Horizontal edges have $s=0$; vertical edges do not form positive-length graph segments.

Let $(x_*,y_*)$ be the junction of the steep and shallow portions, allowing either portion to be empty. The slope of $x+f(x)$ is negative on the steep portion and positive on the shallow portion. Points beyond $t$ cannot improve on $(t,0)$. Therefore $G=x_*+y_*$. Summing vertical descent on the first portion and the second portion gives
\[
 x_*\le\chi(h-y_*)\le\chi(W_1-y_*),\qquad
 y_*\le\chi(t-x_*)\le\chi(W_0-x_*).
\]
Adding proves $(1+\chi)G\le\chi(W_0+W_1)$.

If $P$ is a singleton, $G=0$. For a line segment, an axis direction again has zero gap. A negative slope in mean coordinates attains $(L_0,U_1)$ at one endpoint. For a positive slope, the deficit segment connects $(0,W_1)$ and $(W_0,0)$, so $G=\min(W_0,W_1)$. At a relative interior point some nonzero defining inequality must be tight, or a two-dimensional feasible ball would exist. Its normal is orthogonal to the segment. This normal is a cross-comparison direction, so $W_1/W_0\le\chi$ or $W_1/W_0\ge1/\chi$. Either relation implies $(1+\chi)G\le\chi(W_0+W_1)$. These cases exhaust all dimensions.

Since $P\subseteq[0,1]^2$, $W_0+W_1\le2$. Interchanging the two deficit directions gives the same argument for the lower gap $\min_P(m_1-m_0)-(L_1-U_0)$: the absolute normal-ratio restrictions are unchanged.

\subsection{Optimality and the error tolerance}
Proposition~\ref{prop:sharp-coefficient} attains the width-weighted
coefficient for every $0<\chi\le1/29$ in the primitive-envelope class.
Its positive cells and nonvanishing $W_0$ show a first-order loss near
alignment. The coefficient describes worst-case loss in this primitive-envelope
class; the symmetric JSM imposes additional parametric restrictions.

If a user tolerates an absolute endpoint error $\varepsilon>0$, it suffices to check $\chi(W_0+W_1)/(1+\chi)\le\varepsilon$. A design-only sufficient rule is $2\chi/(1+\chi)\le\varepsilon$; for $0<\varepsilon<2$ this is $\chi\le\varepsilon/(2-\varepsilon)$. These statements apply to zero drift and to this pair-selection model.

\subsection{The uniform coefficient is attained near alignment}
\label{app:sharp-mismatch-coefficient}

The width-dependent coefficient in \eqref{eq:uniform-gap} is optimal in the primitive-envelope class. For each $0<\chi\le1/29$, a strictly positive three-environment observation collection attains equality. The construction uses asymmetric primitive absolute envelopes.

\begin{proposition}[Exact saturation of the mismatch bound]
\label{prop:sharp-coefficient}
For $0\le\zeta\le1/10$, take all three propensities equal to $1/2$ and envelopes
\[
 (\ell_0,u_0)=(1/3,2/3),\quad
 (\ell_1,u_1)=\left(\frac1{3-\zeta},\frac2{3-\zeta}\right),\quad
 (\ell_2,u_2)=\left(\frac{1-\zeta}{3-\zeta},\frac{2-\zeta}{3-\zeta}\right).
\]
Write each observation table as $p_e=\left(\begin{smallmatrix}1/2-a_e&a_e\\1/2-b_e&b_e\end{smallmatrix}\right)$, where
\[
\begin{array}{c|cc}
 e&a_e&b_e\\\hline
 0&(4-\zeta)/20&2/15\\
 1&9(1-\zeta)/\{20(3-\zeta)\}&(4+3\zeta)/\{20(3-\zeta)\}\\
 2&3(5+\zeta)/\{20(3-\zeta)\}&(8+15\zeta-3\zeta^2)/\{40(3-\zeta)\}
\end{array}
\]
All cells are positive. With $L=3(5+\zeta)/40$, $X=9/20$, $Y=1/5$, and $\chi=\zeta/(3-\zeta)$, the complete sharp mean region is
\[
 P=\operatorname{conv}\{(L,Y),(X,Y),(X,Y+\chi(X-L))\}.
\]
Consequently $W_0=X-L$, $W_1=\chi W_0$, and
\begin{equation}\label{eq:sharp-coefficient-equality}
 G=\chi W_0=\frac{\chi}{1+\chi}(W_0+W_1).
\end{equation}
At $\zeta=0$ the region is a horizontal segment and the gap vanishes. As $\zeta\downarrow0$, the observations stay strictly positive and the envelopes approach the common interior envelope $(1/3,2/3)$.
\end{proposition}

\begin{proof}
The slope parameters are $\alpha=(1,1-\zeta,1)$ and $\beta=(2,2,2-\zeta)$. Substitution in \eqref{eq:mismatch-size} gives $\chi=\zeta/(3-\zeta)$. Four inequalities from the exact fiber representation give
\[
 m_0\ge L,\quad m_1\ge Y,\quad m_0\le X,\quad
 -\zeta m_0+(3-\zeta)m_1\le-\zeta L+(3-\zeta)Y.
\]
They are, respectively, the environment-two lower $m_0$ bound, the environment-zero lower $m_1$ bound, $B_0\le C_1$, and $A_2\le C_1$. Their intersection is exactly the displayed triangle. It remains to verify that its three vertices satisfy every other restriction, so the outer triangle is the full model's attainable region.

For each of the three proposed vertices $(x,y)$, put
\[
 c=\frac{13-6\zeta}{20}-(1-\zeta)x-y,\qquad
 q=(1-x-y+c,y-c,x-c,c).
\]
In vertex order $(L,Y),(X,Y),(X,Y+\chi(X-L))$, the three dependence
coordinates are respectively
\[
 c_0=\frac{3(1+\zeta^2)}{40},\qquad
 c_1=\frac{3\zeta}{20},\qquad
 c_2=\frac{3\zeta(5-\zeta)}{40(3-\zeta)}.
\]
Write $a_e=p_e(0,1)$ and $b_e=p_e(1,1)$, and choose
\[
 k_0=2c/3,\qquad k_1=b_1-\frac{y-c}{3-\zeta},\qquad
 k_2=\ell_2c.
\]
The selection masses are
\[
 z_e=(1/2-b_e-x+a_e+k_e,\ b_e-k_e,\ x-a_e-k_e,\ k_e).
\]
These expressions reproduce all four cells in every environment
and normalize both $q$ and the treatment mass exactly.

The score inequalities can be checked by substitution without
optimization. After putting denominators in terms of $3-\zeta>0$,
every nonzero type mass and score slack factors into positive
constants, powers of $\zeta$, linear factors nonnegative on
$[0,1/10]$, and the following positive quadratic factors:
\[
\begin{gathered}
1+\zeta^2,\ (1-\zeta)^2,\ 5-3\zeta^2,\ 1+4\zeta-\zeta^2,\ 4+\zeta-\zeta^2,\\
8-29\zeta+3\zeta^2,\ 20-3\zeta+3\zeta^2,\ 14+3\zeta+3\zeta^2,\\
12-3\zeta-\zeta^2,\ 9+15\zeta-5\zeta^2.
\end{gathered}
\]
The linear factors with a negative slope are among
\[
\begin{gathered}
 1-10\zeta,\ 1-\zeta,\ 2-\zeta,\ 3-\zeta,\ 5-\zeta,\ 9-\zeta,\ 21-\zeta,\\
 4-3\zeta,\ 6-5\zeta,\ 5-23\zeta,\ 4-13\zeta,\ 23-3\zeta,\ 33-13\zeta.
\end{gathered}
\]
Their nonnegativity is immediate at $\zeta\leq1/10$. For example,
$8-29\zeta+3\zeta^2\geq51/10$ and
$12-3\zeta-\zeta^2\geq1169/100$; the other quadratic factors have equally
direct bounds. Thus $q\geq0$ and $\ell_eq\leq z_e\leq u_eq$
throughout the real parameter interval.

Consequently all three vertices have original-model causal
realizations. Convexity makes their entire triangle feasible.
The four valid defining rows already contain the full sharp mean
region in that triangle, proving exact equality and all claimed
extrema.

Since $\chi<1$, the upper ATE is attained at $(L,Y)$ and equals $Y-L$. The marginal subtraction upper is $Y+\chi(X-L)-L$. Their difference gives \eqref{eq:sharp-coefficient-equality}. The three tables have positive entries throughout the interval by their explicit formulas. Finally $X-L=3(1-\zeta)/40$ does not vanish as $\zeta\downarrow0$. Every sufficiently small positive $\chi$ occurs under $\zeta=3\chi/(1+\chi)$, proving that no uniformly smaller coefficient of $W_0+W_1$ works at those mismatch values.
\end{proof}

The continuum polynomial check and exact rational original-mass reconstructions are retained in \path{research19/code/verify_sharp_mismatch.py} and \path{research19/results/sharp_mismatch_certificate.json}. An independent reconstruction uses the explicit original mass constraints and has 96 continuum sign certificates in \path{research19/results/independent_sharp_mismatch_review.json}. These records verify the implementation; the four defining inequalities and the nonnegative slack calculation establish the theorem over the entire real parameter interval. Complementing both binary outcomes exchanges the two ATE endpoints while leaving the envelope design unchanged, and therefore gives the same optimal coefficient for the lower composition gap. The matching construction establishes optimality of the width-dependent coefficient in the primitive-envelope class. Its envelopes are asymmetric; the symmetric JSM is a further restricted subclass.

\section{Robustness of the universal composition boundary}
\label{sec:robust-universal-failure}

The construction in Section~\ref{sec:universal-necessity} gives more
than an isolated counterexample. Its strict local geometry implies
failure on an open set of observed laws, and the failure persists
under sufficiently small positive drift budgets. This is a
robustness consequence of that construction, rather than a separate
characterization of universal composition.

Fix $K\geq2$ environments with target $0$, binary treatment and
outcome, and primitive restrictions
\[
 0<\ell_e<\pi_e<u_e<1\qquad(e=0,\ldots,K-1).
\]
Assume that at least two pairs $(\ell_e,u_e)$ differ. The envelopes
and propensities remain fixed throughout this section. Let
$\mathcal P_\pi^+$ denote the strictly positive observed tables
satisfying
\[
 p_e(0,0)+p_e(0,1)=1-\pi_e,\qquad
 p_e(1,0)+p_e(1,1)=\pi_e.
\]
Equip this set with the topology relative to these affine row-sum
constraints. At a star-TV budget
$\boldsymbol\delta=(\delta_e)_{e=1}^{K-1}\geq0$, the model has laws
$q_e$, selection masses $z_e$, and restrictions
\[
 \ell_e q_{et}\leq z_{et}\leq u_e q_{et},\qquad
 \operatorname{TV}(q_e,q_0)\leq\delta_e.
\]
Let $L_0(p,\boldsymbol\delta)$ and
$U_1(p,\boldsymbol\delta)$ be the lower target mean under treatment
zero and upper target mean under treatment one, respectively, both
optimized over this \emph{full} model. Let
$U_\tau(p,\boldsymbol\delta)$ be its sharp target ATE upper.

\begin{theorem}[Robust failure for every mismatch]\label{thm:robust}
There exist $p^*\in\mathcal P_\pi^+$, $r_*>0$, $d_*>0$, and
$g_*>0$ such that the model is nonempty and
\begin{equation}\label{eq:robust-gap}
 U_1(p,\boldsymbol\delta)-L_0(p,\boldsymbol\delta)
       -U_\tau(p,\boldsymbol\delta)\geq g_*
\end{equation}
whenever
\[
 p\in\mathcal P_\pi^+,\quad
 \|p-p^*\|_\infty<r_*,\quad
 0\leq\delta_e\leq d_*\quad(e=1,\ldots,K-1).
\]
In particular, at zero drift the tables with a strict composition
gap contain a nonempty relatively open set. The same open set has
a strict gap for every sufficiently small positive star-TV budget.
The constants depend on the fixed design and construction; no
uniform lower gap or radius as envelopes approach alignment is
asserted.

\end{theorem}

\begin{proof}
\emph{A strict center at the constructed tables.}
Use the fixed-propensity construction of
Section~\ref{sec:universal-necessity} to obtain positive tables
$p^*$ and a mean point $m^*$ at which one cross-fiber inequality is
active. Every local-rectangle inequality and every other
lower-versus-upper fiber comparison is strict. The active
inequality has the form
\[
 -a(m_0-m_0^*)+b(m_1-m_1^*)\leq0,
 \qquad a,b>0.
\]
For sufficiently small $s>0$, the point
$m^c=m^*+s(a,-b)$ makes this inequality strict while preserving all
the others. Thus $m^c$ lies strictly inside every local rectangle,
and the largest lower fiber endpoint is strictly less than the
smallest upper endpoint. Choose $c$ strictly between them.

For each environment, the low and high group marginals are
strictly between zero and their group masses. Each group's binary
coupling interval therefore has nonempty interior. Since $c$ lies
strictly inside the sum of the two intervals, it can be written as
a sum of points strictly inside the two group intervals. These
choices give two group laws whose four coordinates are positive.
They sum to the same law
\[
 q^c=(1-m_0^c-m_1^c+c,\ m_1^c-c,\ m_0^c-c,\ c)
\]
in every environment, and their reconstructed selection masses
satisfy
\begin{equation}\label{eq:robust-strict-center}
 q_t^c>0,\qquad
 z_{et}^c-\ell_e q_t^c>0,\qquad
 u_e q_t^c-z_{et}^c>0.
\end{equation}
This is strict feasibility for the mass and selection restrictions
in the common-law formulation. The zero-drift inequalities
in an expanded star formulation are necessarily tight at zero;
strictness of those inequalities is not claimed or needed.

\emph{Two strict witnesses with a separating crossed mean.}
At $(p^*,\boldsymbol0)$, write
\[
 L^*=L_0(p^*,\boldsymbol0),\qquad
 H^*=U_1(p^*,\boldsymbol0),\qquad
 U^*=U_\tau(p^*,\boldsymbol0).
\]
The universal necessity theorem gives
$G=H^*-L^*-U^*>0$. Choose feasible witnesses $x^L,x^H$ attaining
$L^*,H^*$; they exist by compactness. Mix each with the strict center
$x^c$ at the same weight $0<\theta<\min\{1/2,G/4\}$:
\[
 \widetilde x^L=(1-\theta)x^L+\theta x^c,\qquad
 \widetilde x^H=(1-\theta)x^H+\theta x^c.
\]
Both mixed witnesses still reproduce $p^*$ and have a common law
across environments within each witness. Every inequality in
\eqref{eq:robust-strict-center} is now strict for both. Their
crossed mean is
\[
 C=\mu_1(\widetilde x^H)-\mu_0(\widetilde x^L)
   =(1-\theta)(H^*-L^*)+\theta\tau(x^c).
\]
As $H^*-L^*$ and $\tau(x^c)$ belong to $[-1,1]$,
\begin{equation}\label{eq:robust-cross-gap}
 g:=C-U^*\geq G-2\theta>G/2>0.
\end{equation}
It is important here to use the full-model coordinate extremizers:
two points from the local tangent construction alone need not have
a crossed mean above the \emph{global} ATE upper.

\emph{An explicit observation correction.}
Write a nearby fixed-propensity table as
\[
 p_e-p_e^*=
 \begin{pmatrix}-d a_e&d a_e\\-d b_e&d b_e\end{pmatrix}.
\]
For each mixed witness keep all its $q_e$ fixed and change the
selection vector in environment $e$ by
\begin{equation}\label{eq:robust-z-correction}
 d z_e=(d a_e-d b_e,\ d b_e,\ -d a_e,\ 0).
\end{equation}
This correction produces the stated change in every observed cell
and satisfies $\sum_t d z_{et}=0$. Moreover
$\|d z_e\|_\infty\leq2\|p-p^*\|_\infty$.
Let $s_*>0$ be the smallest selection slack
$z_{et}-\ell_e q_{et}$ or $u_e q_{et}-z_{et}$ among these two
mixed witnesses and all environments and types. If
$2\|p-p^*\|_\infty<s_*$, both corrected witnesses remain feasible.
Their positive $q_e$ are unchanged; the strict score bounds imply
$0<z_{et}<q_{et}$. Their within-witness drift remains zero, so they
are feasible for every nonnegative star budget. Their target means
are unchanged, giving
\begin{equation}\label{eq:robust-oracle-witness}
 U_1(p,\boldsymbol\delta)-L_0(p,\boldsymbol\delta)\geq C.
\end{equation}

\emph{Upper semicontinuity completes the persistence argument.}
All latent masses lie in a fixed compact cube and the constraints are
continuous in $(p,\boldsymbol\delta)$. If maximizing models along a
sequence approaching $(p^*,0)$ had objective above $U^*+g/2$, a
convergent subsequence would give a zero-drift model at $p^*$ with
objective above $U^*$, a contradiction. Hence sufficiently small
$r_*,d_*>0$ guarantee $U_\tau(p,\boldsymbol\delta)<U^*+g/2$.
Let $p_{\min}=\min_{e,a,y}p_e^*(a,y)>0$. Choose $r_*$ also below $p_{\min}/2$ and $s_*/4$ so the corrected
strict witnesses remain positive and feasible. Combining this upper
with \eqref{eq:robust-oracle-witness} and $C-U^*=g$ proves
\eqref{eq:robust-gap} with $g_*=g/2$, including nonemptiness.
\end{proof}

Appendix~\ref{app:computable-radius} supplies a single quantitative
version using a common-law dual and source coefficient ranges.
The constants depend on the fixed design and data; they can vanish
as mismatch vanishes. These observation neighborhoods fix propensities and envelopes and
certify persistence of the composition gap. The paired sign
neighborhoods are established separately in Appendix~\ref{app:decision-information}.

\section{A computable neighborhood certificate}
\label{app:computable-radius}

The robust-failure conclusion can be made quantitative using one linear-programming dual bound \cite{bertsimas} and two feasible witnesses. This certificate keeps the primitive score envelopes and treatment propensities fixed. It allows observed-table perturbations within the fixed-propensity slice and positive source-drift budgets. The radii are sufficient bounds, not maximal robustness radii.

\paragraph{Matrix convention.}
At zero drift write $x=(q,z_0,\ldots,z_{K-1})\geq0$, with types ordered $00,01,10,11$. The nominal constraints are
\[
Ex=f,\qquad Ax\leq0,
\]
where $E$ contains $\mathbf1^\top q=1$ and the four observed-cell equations per environment, while $A$ contains the rows $\ell_e q-z_e\leq0$ and $z_e-u_e q\leq0$. Let $c^\top x=-\tau_0$. A dual-feasible triple satisfies
\begin{equation}\label{eq:radius-dual}
c=E^\top\nu+A^\top\lambda+r,
\qquad \lambda\leq0,\quad r\geq0,
\end{equation}
with unrestricted $\nu$. It certifies the upper bound $\bar U=-\nu^\top f$ on the ATE; dual optimality is not required.

Let $\nu_e=(\nu_{e00},\nu_{e01},\nu_{e10},\nu_{e11})$ denote the four observation multipliers. For source $e>0$, let $E_{e,q}$ and $A_{e,q}$ be the columns corresponding to $q$ in that environment's observation and selection rows. Define
\begin{equation}\label{eq:radius-constants}
g_e=E_{e,q}^{\top}\nu_e+A_{e,q}^{\top}\lambda_e,
\quad D_e^\dagger=\max_tg_{et}-\min_tg_{et},
\quad
B_{\rm obs}=\sum_{e=0}^{K-1}\bigl(|\nu_{e01}-\nu_{e00}|+|\nu_{e11}-\nu_{e10}|\bigr).
\end{equation}
Total variation controls differences of bounded-function expectations \cite{gibbs}. Thus the drift penalties can be computed directly from a common-law dual without solving an additional drift-constrained program.

\paragraph{Certificate.}
Suppose $x^L,x^H$ are nominal feasible witnesses and $x^c$ is a nominal feasible witness with strictly positive type masses and strictly positive selection slacks. Put
\[
C_0=\mu_1(x^H)-\mu_0(x^L),\qquad G_{\rm cert}=C_0-\bar U>0,
\qquad t_c=\mu_1(x^c)-\mu_0(x^c).
\]
The witnesses need not optimize the marginal objectives. If they do and the dual is optimal, $G_{\rm cert}$ is the nominal oracle-composition gap. Set
\[
\omega=\frac{G_{\rm cert}}{2(C_0-t_c)},\qquad
\widetilde x^L=(1-\omega)x^L+\omega x^c,\qquad
\widetilde x^H=(1-\omega)x^H+\omega x^c.
\]
Put $\Delta=G_{\rm cert}/2$, the mixed crossed mean minus $\bar U$. Let $s>0$ be the smallest selection slack across these two mixed witnesses and let $p_{\min}>0$ be the smallest nominal observed cell. Consider perturbations
\[
p'_e-p_e=(-\epsilon_{e0},\epsilon_{e0},-\epsilon_{e1},\epsilon_{e1}),\qquad |\epsilon_{e0}|,|\epsilon_{e1}|\leq\rho,
\]
where cells are ordered $00,01,10,11$. Write $U_\tau(p',\delta)$ for the full-model ATE upper endpoint and $B_{\rm oracle}(p',\delta)$ for the upper endpoint from its coordinate ranges. If
\[
0\leq\rho\leq\min\{p_{\min}/2,s/4\},\qquad \delta_e\geq0,
\]
then the perturbed model is nonempty and
\begin{equation}\label{eq:computable-radius-gap}
B_{\rm oracle}(p',\delta)-U_\tau(p',\delta)
\geq\Delta-B_{\rm obs}\rho-\sum_{e>0}D_e^\dagger\delta_e.
\end{equation}
In particular, with a zero denominator interpreted as $+\infty$, choose
\begin{equation}\label{eq:computable-radius-choice}
\rho_* =\min\left\{\frac{p_{\min}}2,\frac s4,\frac{G_{\rm cert}}{8B_{\rm obs}}\right\},
\qquad
D_* =\min\left\{1,\frac{G_{\rm cert}}{8\sum_{e>0}D_e^\dagger}\right\}.
\end{equation}
Every table satisfying the displayed perturbation bounds with $\rho\leq\rho_*$, and every source budget $0\leq\delta_e\leq D_*$, has an oracle-composition gap of at least $G_{\rm cert}/4>0$.

\begin{proof}
Dual feasibility gives $t_c\leq\bar U<C_0$, hence $0<\omega\leq1/2$. The mixed witnesses have crossed-mean value
\[
\widetilde C=\mu_1(\widetilde x^H)-\mu_0(\widetilde x^L)
=(1-\omega)C_0+\omega t_c=\bar U+G_{\rm cert}/2.
\]
This deliberate mixing spends half of the certified nominal gap to obtain strictly feasible witnesses; it does not assert that the perturbed sharp gap loses half its value.

For each mixed witness, keep its common $q$ fixed and change its environment-$e$ selection vector by
\[
\Delta z_e=(\epsilon_{e0}-\epsilon_{e1},\ \epsilon_{e1},\ -\epsilon_{e0},\ 0).
\]
The four observation equations change by exactly $(-\epsilon_{e0},\epsilon_{e0},-\epsilon_{e1},\epsilon_{e1})$. Each selection mass changes by at most $2\rho$, so every selection slack remains at least $s-2\rho\geq s/2$. The corrected witnesses have unchanged means and zero drift and are therefore feasible for all the stated nonnegative drift budgets. The perturbed oracle upper endpoint is at least $\widetilde C$.

For an arbitrary perturbed feasible model, separate the common $q$ into $q_0,\ldots,q_{K-1}$. Assign the common normalization multiplier and the common-$q$ reduced-cost term in \eqref{eq:radius-dual} to $q_0$. Expanding that coefficient identity gives
\begin{align*}
-\tau_0={}&\nu_{\rm norm}+\sum_e\nu_e^\top p'_e
 +\sum_e\lambda_e^\top A_e(q_e,z_e)
 +r_q^\top q_0+\sum_e r_{z_e}^\top z_e
 +\sum_{e>0}g_e^\top(q_0-q_e).
\end{align*}
The selection and reduced-cost terms are nonnegative. Since $q_0-q_e$ has total mass zero,
$g_e^\top(q_0-q_e)\geq-D_e^\dagger\delta_e$. The observation perturbation changes $\nu^\top f$ by at most $B_{\rm obs}\rho$. It follows that
\[
U_\tau(p',\delta)\leq\bar U+B_{\rm obs}\rho+\sum_{e>0}D_e^\dagger\delta_e.
\]
Subtracting this from $\widetilde C$ proves \eqref{eq:computable-radius-gap}. The choices in \eqref{eq:computable-radius-choice} allocate at most $G_{\rm cert}/8$ each to observation perturbation and drift, giving the final $G_{\rm cert}/4$ guarantee.
\end{proof}

\paragraph{Verified instance.}
For the two-environment construction with $\Gamma=2$ and $\pi=(1/4,3/4)$, the exact nominal upper gap is $G_{\rm cert}=13/320$. Rationally verified witnesses and a dual give
\[
B_{\rm obs}=139/30,\qquad D_1^\dagger=1,\qquad
\rho_* =\frac{8619}{22328320},\qquad D_* =\frac{13}{2560}.
\]
Consequently the rounded-down radii $\rho=0.000386$ and $D=0.005078$ certify, simultaneously for every allowed table perturbation and every source budget $0\leq\delta_1\leq D$, a composition gap of at least
\[
G_{\rm cert}/4=13/1280=0.01015625.
\]
The supplied code reconstructs the strict center over the rationals and checks both mixed witnesses, the observation correction, the dual decomposition, and the radius inequalities exactly. Its saved certificates also cover four further nonaligned constructions, including two near alignment and a three-environment model. The near-aligned examples have correspondingly smaller certified radii. These certificates bound population composition loss at externally specified drift budgets. Appendix~\ref{app:revision-evidence} gives the separate finite-sample confidence construction.

\section{Composition loss under aligned envelopes and drift}
\label{sec:drift-rectangularity-boundary}
The following construction quantifies the role of response-law stability in Theorem~\ref{thm:universal}. With aligned score limits, allowing response laws to differ within a TV budget creates a composition gap linear in that budget. The observed tables remain strictly positive throughout the stated parameter range.

Fix any $\Gamma\in[2,4]$, $\ell=1/(\Gamma+1)$, $u=\Gamma/(\Gamma+1)$,
$L=9/40$, $M=9/20$, $T=1-L-M=13/40$, and $c=3/20$.
Both propensities are $1/2$. Take target observations
$a_0=1/10$, $b_0=1/4$, and source observations
$a_1=uL$, $b_1=1/2-uM$. All four cells in both environments are
strictly positive. For every $0\leq\delta\leq1/100$ the full sharp
extrema are
\[
 L_0=L-\delta,\qquad U_1=1-M+\delta,\qquad
 U_{\rm ATE}=T+(1+1/\Gamma)\delta,
 \qquad G_\delta^+=(1-1/\Gamma)\delta.
\]
Thus a uniformly $o(\Delta)$ drift term is impossible, including at
alignment. At fixed $\Gamma=4$ the displayed family has exact gap $3\delta/4$.

\begin{proof}
The source cells $a_1=uL$ and $p_1(1,0)=uM$ give its bounds
$m_{0,1}\geq L$, $m_{1,1}\leq1-M$, hence the first two target bounds
by TV. Also
\[
 u(L+M)=a_1+p_1(1,0)
  =q_{1,10}+q_{1,11}+z_{1,00}-z_{1,11}
  \leq q_{1,10}+u(q_{1,00}+q_{1,11}).
\]
Rearranging gives
$\phi(q_1)=q_{1,01}-q_{1,10}/\Gamma\leq T$.
Its coefficient range is $1+1/\Gamma$. Thus
$\operatorname{ATE}(q_0)\leq\phi(q_0)\leq T+(1+1/\Gamma)\delta$.
The following original-mass witnesses attain these three bounds.
In every case $q_0-q_1=(0,\delta,-\delta,0)$, so TV equals $\delta$.

For the lower $m_0$ witness use
\[
 q_1=(M-L/\Gamma,T+L/\Gamma,L,0),\quad
 q_0=q_1+(0,\delta,-\delta,0),
\]
\[
 z_1=(uM-\ell L,b_1,\ell L,0),\quad
 z_0=(1/8+\delta,1/4,1/8-\delta,0).
\]
For the upper $m_1$ witness use
\[
 q_1=(M-c-\Gamma\delta,T+c/\Gamma+\delta,
       c+\Gamma\delta,L-c/\Gamma-\delta),\quad
 q_0=q_1+(0,\delta,-\delta,0).
\]
With $k=u q_{0,11}$, put
\[
 z_1=(u q_{1,00},b_1-\ell q_{1,11},u q_{1,10},\ell q_{1,11}),
\]
\[
 z_0=(7/20-q_{0,10}-\ell q_{0,11},1/4-k,
       q_{0,10}+\ell q_{0,11}-1/10,k).
\]
For the upper ATE witness use
\[
 q_1=(M-\delta/\Gamma,T+\delta/\Gamma,\delta,L-\delta),\quad
 q_0=q_1+(0,\delta,-\delta,0),
\]
\[
 z_1=(u q_{1,00},b_1-\ell q_{1,11},\ell\delta,\ell q_{1,11}),
 \qquad z_0=(1/4,1/8+\delta,0,1/8-\delta).
\]
Normalization and observations hold identically. Throughout the
displayed parameter rectangle every $q$ is nonnegative and
$\ell q\leq z\leq uq$, as certified by the exact polynomial signs
described below. Substitution attains each claimed endpoint.
The marginal-box endpoint is $T+2\delta$, proving the gap formula.
\end{proof}

\paragraph{Exact check.}
The analytic identities above give the bounds and attaining objectives. The remaining mass and score inequalities over all real $(\Gamma,\delta)\in[2,4]\times[0,1/100]$ have exact rational Bernstein sign certificates in \path{research22/theory/results/aligned_drift_lower_certificate.json}, reconstructed by \path{research22/theory/code/verify_aligned_drift_lower.py}. Outcome complementation supplies the lower-gap example. The family proves that a uniform drift bound must retain a first-order term; its coefficient is exact for the displayed tables.

\section{Coverage and the frozen sampling protocol}
\label{app:revision-evidence}

Set inference includes asymptotic confidence regions \cite{cht,romanoshaikh}, moment-inequality methods \cite{andrewssoares}, and calibrated projection intervals \cite{kaido}. Here we enclose observable probabilities and project the resulting causal feasible set to cover the whole population ATE interval.

\paragraph{Coverage of the common outer model.}
For $K$ independent multinomial count tables, construct two-sided Clopper--Pearson intervals \cite{cp} for four cell events, treatment, and the unions $\{(0,0),(1,0)\}$ and $\{(0,0),(1,1)\}$. Each interval has error probability at most $\alpha/(7K)$; independence among these seven events is not needed. A union bound yields an event $\mathcal E$ of probability at least $1-\alpha$ on which all true event probabilities lie in their intervals. We constrain the induced observation probabilities of the original-mass model to these intervals. For a common JSM, use $[\ell(\pi_e^-;\Gamma),u(\pi_e^+;\Gamma)]$; for primitive envelopes, keep those specified absolute limits. Monotonicity of both JSM endpoint functions ensures that the true envelope is contained in this outer envelope.

On $\mathcal E$, every population-compatible causal model belongs to the same outer model. All its coordinate and ATE extrema therefore contain their population counterparts, and its ATE interval contains the whole population identified interval. Declaring negativity only when this model is nonempty and its upper is negative has false-declaration probability at most $\alpha$ under the composite null that the population identified interval admits a nonnegative ATE. Empty models make no declaration. This coverage guarantee follows from outer-region inclusion and the union bound. The numerical implementation uses floating-point quantiles and LP solutions with residual checks; it does not provide outward-certified arithmetic.

\paragraph{Frozen sampling of the population sign families.}
\path{research23/PROTOCOL.md} was fixed before new sampling. Proposition~\ref{thm:insufficiency} uses $t=\pm1$, $\Gamma=2$, $\pi=(1/4,3/4)$, and budgets $0,1/320$. Its four settings are retained, including the positive-$t$ controls. The mirror settings fix $(l,u,p)=(1/7,2/5,1/4)$, $\rho=5/12$, $\epsilon=\epsilon_0/2$ from the exact finite-minimum construction, and $t=\pm\epsilon/2$, with zero drift. Their primitive selection envelopes are fixed inputs, while their observed propensities are constrained by the confidence intervals.

Each of the six settings uses 200 trials at each $n=3000,10000,30000$ per environment, with seed 2026100423. The declaration threshold is $-10^{-7}$. All 3600 trial pairs (7200 environment count tables), six numerical extrema, residuals and coverage indicators are preserved in \path{research23/abstract_family_raw.jsonl}; settings and software versions are in \path{research23/sampling_manifest.json}. The summaries are in \path{research23/abstract_family_summary.json}. Every setting has $0/200$ negative declarations for either method at each sample size. A pointwise 95\% binomial Monte Carlo interval for $0/200$ is $[0,.01828]$. These are pointwise Monte Carlo intervals for declaration probabilities at the fixed settings under the reported confidence procedure.

\paragraph{Verification.}
\path{research23/verify_experiments.py} assembles the absolute-TV constraints independently and checks the six extrema of every saved trial pair. Counts, maximum discrepancies and primal--dual residuals are in \path{research23/experiment_verification.json}. The two independently assembled constraint systems use the same SciPy/HiGHS backend. Their agreement verifies constraint assembly, with numerical accuracy assessed through the reported primal--dual residuals.

Historical mechanism and confidence diagnostics remain in the accompanying code archive with their protocols and saved numerical rows. Its README indexes these records separately from the frozen experiment reported here.

\clearpage
\section{Calculation, scope and reproduction}
\subsection{Calculate the supplied observation table}
\label{app:workflow-records}
Run the population diagnostic from the reproducibility package root:
\begin{verbatim}
python research24/diagnose.py \
  research24/examples/moderate_propensity_Gamma2.json \
  --output diagnosis.json
python research24/verify_workflow.py
\end{verbatim}
The JSON input specifies treatment-row, outcome-column tables and primitive \texttt{ell,u} or $\texttt{gamma}>1$. Fraction strings preserve exact inputs; decimals are certified as their stated values. The routine uses aligned formulas or certified three-variable bases, checks both ATE corners and returns rational endpoints, causal witnesses or conflicting pieces. Failed basis certification returns no exact conclusion. Count-table inference and positive drift use the separate full model.

Eight inputs combine seven archived boundary examples with Section~\ref{sec:workflow}'s unaligned exact-composition table. Complementation, environment reversal and duplication give 32 checks: 192 endpoints agree with the original-mass LP to $4.285\times10^{-16}$, and 120 primal--dual certificates are rationally verified. Original causal witnesses and eight invalid-input rejections are also checked. \path{research24/workflow_verification.json} and \path{research24/examples} retain the records; \path{research24/README.md} specifies certificate failure handling. Fiber separation and ATE loss are reported separately.

\subsection{Relation to prior work}
\label{app:relation-scope}
Binary-IV reviews organize bounds by assumptions \cite{swanson}; instrumentality tests \cite{bonet} and observable constraints \cite{ramsahai} address compatibility, while graphical relaxations support causal bounding \cite{silvaevans}. Automated polynomial optimization handles general discrete models \cite{duarte}. Song and coauthors characterize categorical IV compatibility and variation dependence \cite{song}. Our finite assignment bounds restrict that parent model: feasible $(q,z_e)$ satisfy its characterization, but a parent-compatible $q$ may fail selection constraints.

Richardson--Robins' binary IV mean region is rectangular \cite{rr}. At Proposition~\ref{thm:insufficiency}'s $t=\delta=0$ tables, that parent region is $[3/8,5/8]\times[23/80,43/80]$; the pair-selection submodel is a hexagon with ATE upper $0$ and marginal upper $11/480$ (Appendix~\ref{app:uniform-mismatch}). Theorem~\ref{thm:universal}'s outcome-specific comparison locates the boundary in restrictions on the whole response pair.

Sensitivity analysis varies assumptions about hidden covariates \cite{rrsens} or assignment bias within matched sets \cite{rosenbaum87}. Recent work studies robust bound estimation \cite{dgk,tan24}, additional outcome restrictions \cite{enhanced}, and generalized models \cite{frauen}. Single-environment simultaneous attainment \cite{dg,mpr} is our starting point; a shared response law adds cross-environment coupling.

Missing-data bounds \cite{horowitz} concern incomplete observation, whereas corner loss here can occur with fully observed tables. Fibers test composition for a supplied zero-drift table; our examples establish information loss, while failure frequencies across application populations remain unmeasured.

\subsection{Verification and environment}
\label{app:reproducibility}
The accompanying reproducibility README gives execution commands, dependency pins and a map of the code, numerical records and experiment protocols. Rational audits reconstruct the corner examples, sign-family witnesses and perturbation certificates. Symbolic interval checks cover the primitive matching triangle and aligned-drift masses. The analytic proofs establish the universal statements. Primal--dual residuals document numerical accuracy of the confidence-LP calculations.

CPU software versions are Python 3.12.14, NumPy 2.3.5, SciPy 1.18.1, SymPy 1.14.0 and Matplotlib 3.10.8, pinned in \path{research23/requirements.txt}. No GPU or model training is used. CPU and memory specifications were not recorded. Run producers in a copy; they overwrite named records. Dependency provenance: \path{THIRD_PARTY_NOTICES.md}.

\end{document}